\documentclass[a4paper,UKenglish,numberwithinsect,cleveref,thm-restate]{lipics-v2021}
\usepackage{xspace,stmaryrd,import,mathtools}
\hypersetup{hypertexnames=false} 
\newcommand{\U}{\ensuremath{\mathsf{U}}\xspace}
\newcommand{\mso}{\textup{MSO}\xspace}

\newcommand{\msou}{\textup{MSO+}\U}
\newcommand{\wmsou}{\textup{WMSO+}\U}
\newcommand{\wmsouk}{\textup{WMSO+}$\U_{\mathsf{tup}}$\xspace}
\renewcommand{\phi}{\varphi}
\renewcommand{\L}{\mathsf{L}}
\newcommand{\R}{\mathsf{R}}
\renewcommand{\S}{\mathsf{S}}
\newcommand{\D}{\mathsf{D}}
\newcommand{\F}{\mathsf{F}}
\newcommand{\G}{\mathsf{G}}
\newcommand{\X}{\mathsf{X}}
\newcommand{\Y}{\mathsf{Y}}
\newcommand{\Z}{\mathsf{Z}}
\renewcommand{\a}{\mathsf{a}}
\renewcommand{\b}{\mathsf{b}}
\renewcommand{\c}{\mathsf{c}}
\renewcommand{\d}{\mathsf{d}}

\newcommand{\g}{\mathsf{g}}

\newcommand{\x}{\mathsf{x}}
\newcommand{\Vv}{\mathcal{V}}
\newcommand{\Ff}{\mathcal{F}}
\newcommand{\Pp}{\mathcal{P}}
\newcommand{\Gg}{\mathcal{G}}
\newcommand{\Ll}{\mathcal{L}}
\newcommand{\Tt}{\mathcal{T}}
\newcommand{\Uu}{\mathcal{U}}
\newcommand{\AlA}{\mathbb{A}}
\newcommand{\AlB}{\mathbb{B}}
\newcommand{\Nat}{\mathbb{N}}
\newcommand{\set}[1]{\{#1\}}
\newcommand{\symb}[1]{\langle #1\rangle}
\newcommand{\restr}{{\upharpoonright}}
\newcommand{\dom}{\mathrm{dom}}
\renewcommand{\epsilon}{\varepsilon}
\newcommand{\child}{\curlywedgedownarrow} 
\newcommand{\efin}{{\exists_{\mathsf{fin}}}}
\newcommand{\lamdots}{.\cdots{}.}
\newcommand{\Pht}[1]{\mathit{Typ}_{#1}}
\newcommand{\pht}[3]{\llbracket #2 \rrbracket_{#1}^{#3}}
\newcommand{\SUP}{\mathit{SUP}}
\newcommand{\nd}{\mathsf{nd}}
\newcommand{\refl}{\mathsf{refl}}
\newcommand{\rednd}{\to_\nd}
\newcommand{\true}{\mathsf{tt}}
\newcommand{\false}{\mathsf{ff}}
\newcommand{\Comp}{\mathit{Comp}}
\newcommand{\interval}[1]{[#1]}
\Crefname{equation}{Equality}{Equalities}
\newtheorem{fact}[theorem]{Fact}
\Crefname{fact}{Fact}{Facts}
\Crefname{lemma}{Lemma}{Lemmata}

\author{Anita Badyl}{Institute of Informatics, University of Warsaw, Poland}{}{}{}
\author{Paweł Parys}{Institute of Informatics, University of Warsaw, Poland}{parys@mimuw.edu.pl}{https://orcid.org/0000-0001-7247-1408}{}
\authorrunning{A. Badyl and P. Parys}
\Copyright{Anita Badyl and Paweł Parys}
\funding{Work supported by the National Science Centre, Poland (grant no.\@ 2016/22/E/ST6/00041).}
\relatedversion{This is an extended version of a paper published at the CSL 2024 conference.} 
\relatedversiondetails{Published Version}{https://doi.org/10.4230/LIPIcs.CSL.2024.2} 
\nolinenumbers
\hideLIPIcs

\title{Extending the WMSO+U Logic With Quantification Over Tuples}



\EventEditors{Aniello Murano and Alexandra Silva}
\EventNoEds{2}
\EventLongTitle{32nd EACSL Annual Conference on Computer Science Logic (CSL 2024)}
\EventShortTitle{CSL 2024}
\EventAcronym{CSL}
\EventYear{2024}
\EventDate{February 19--23, 2024}
\EventLocation{Naples, Italy}
\EventLogo{}
\SeriesVolume{288}
\ArticleNo{2}

\begin{document}

\maketitle

\begin{abstract}
	We study a new extension of the weak \mso logic, talking about boundedness.
	Instead of a previously considered quantifier \U, expressing the fact that there exist arbitrarily large finite sets satisfying a given property,
	we consider a generalized quantifier $\U$, expressing the fact that there exist tuples of arbitrarily large finite sets satisfying a given property.
	First, we prove that the new logic \wmsouk is strictly more expressive than \wmsou.
	In particular, \wmsouk is able to express the so-called simultaneous unboundedness property,
	for which we prove that it is not expressible in \wmsou.
	Second, we prove that it is decidable whether the tree generated by a given higher-order recursion scheme satisfies a given sentence of \wmsouk.
\end{abstract}

\section{Introduction}

In the field of logic in computer science, one of the goals is to find logics that, on the one hand, have decidable properties and, on the other hand, are as expressive as possible.
An important example of such a logic is the monadic second-order logic, \mso,
which defines exactly all regular properties of finite and infinite words~\cite{Buchi,Elgot,Trakhtenbrot} and trees~\cite{RabinMSO},
and is decidable over these structures.

A natural question that arises is whether \mso can be extended in a decidable way.
Particular hopes were connected with expressing boundedness properties.
Bojańczyk~\cite{BojanczykU} introduced a logic called \msou, which extends \mso with a new quantifier \U,
with $\U X.\phi$ saying that the subformula $\phi$ holds for arbitrarily large finite sets $X$.
Originally, it was only shown that satisfiability over infinite trees is decidable for formulae where the $\U$ quantifier is only used once and not under the scope of set quantification.
A significantly more powerful fragment of the logic, albeit for infinite words, was shown decidable by Bojańczyk and Colcombet~\cite{bounds} using automata with counters.
These automata  were further developed into the theory of cost functions initiated by Colcombet~\cite{regular-cost-functions}.

The difficulty of \msou comes from the interaction between the \U quantifier and quantification over possibly infinite sets.
This motivated the study of \wmsou, which is a variant of \msou where set quantification is restricted to finite sets (the ``W'' in the name stands for \emph{weak}).
On infinite words, satisfiability of \wmsou is decidable, and the logic has an automaton model~\cite{wmsou-words}.
Similar results hold for infinite trees~\cite{wmsou-trees}, and have been used to decide properties of \textsc{ctl*}~\cite{CarapelleKL13}.
Currently, the strongest decidability result in this line is about \wmsou over infinite trees extended with quantification over infinite paths~\cite{wmsou-path}.
The latter result entails decidability of problems such as the realisability problem for prompt \textsc{ltl}~\cite{KupfermanPitermanVardi09},
deciding the winner in cost parity games~\cite{FijalkowZ12}, or deciding certain properties of energy games~\cite{BrazdilCKN12}.

The results mentioned so far concern mostly the satisfiability problem (is there a model in which a given formula is true?),
but arguably the problem more relevant in practice is the model-checking problem: is a given formula satisfied in a given model?
In a typical setting, the model represents (possible computations of) some computer system, and the formula expresses some desired property of the system, to be verified.
The model is thus usually infinite, although described in a finite way.
In this paper, as the class of considered models we choose trees generated by higher-order recursion schemes,
which is a very natural and highly expressive choice.

Higher-order recursion schemes (recursion schemes in short) are used to faithfully represent the control flow of programs in languages with higher-order functions~%
\cite{Damm82,KNU-hopda,Ong-hoschemes,KobayashiACM}.
This formalism is equivalent via direct translations to simply-typed $\lambda Y$-calculus~\cite{schemes-lY}.
Collapsible pushdown systems~\cite{collapsible} and ordered tree-pushdown systems~\cite{Ordered-Tree-Pushdown} are other equivalent formalisms.
Recursion schemes easily cover finite and pushdown systems, but also some other models such as indexed grammars~\cite{AhoIndexed} and ordered multi-pushdown automata~\cite{OrderedMultiPushdown}.

A classic result, with several proofs and extensions, says that model-checking trees generated by recursion schemes against \mso formulae is decidable:
given a recursion scheme $\Gg$ and a formula $\phi\in\mso$, one can say whether $\phi$ holds in the tree generated by $\Gg$~%
\cite{Ong-hoschemes,collapsible,KobayashiOngtypes,KrivineWS,reflection-selection,ModelSW}.
When it comes to boundedness properties, one has to first mention decidability of the simultaneous unboundedness property
(a.k.a.~diagonal property)~\cite{diagonal-safe,diagonal,types-diagonal-journal}.
In this problem one asks whether, in the tree generated by a given recursion scheme $\Gg$, there exist branches containing arbitrarily many occurrences of each of the labels $a_1,\dots,a_k$
(i.e., whether for every $n\in\Nat$ there exists a branch on which every label from $\set{a_1,\dots,a_k}$ occurs at least $n$ times).
This result turns out to be interesting, because it entails other decidability results for recursion schemes,
concerning in particular computability of the downward closure of recognized languages~\cite{Zetzsche-dc},
and the problem of separability by piecewise testable languages~\cite{sep-piecewise-test}.
Then, we also have decidability for logics talking about boundedness.
Namely, it was shown recently that model-checking for recursion schemes is decidable against formulae from \wmsou~\cite{wmsou-schemes}
(and even from a mixture of \mso and \wmsou, where quantification over infinite sets is allowed but cannot be arbitrarily nested with the \U quantifier~\cite{wmsou-schemes-journal}).
Another paper~\cite{wcmso-safe-schemes} shows decidability of model-checking for a subclass of recursion schemes
against alternating B-automata and against weak cost monadic second-order logic (WCMSO);
these are other formalisms allowing to describe boundedness properties, but in a different style than the \U quantifier.

Interestingly, the decidability of model-checking for \wmsou is obtained by a reduction to (a variant of) the simultaneous unboundedness problem.
On the other hand, it seems that the simultaneous unboundedness property cannot be expressed in \wmsou
(except for the case of a single distinguished letter $a_1$), which is very intriguing.

\subparagraph{Our contribution.}

As a first contribution, we prove the fact that was previously only a hypothesis: \wmsou is indeed unable to express the simultaneous unboundedness property.
Then, we define a new logic, \wmsouk; it is an extension of \wmsou, where the \U quantifier can be used with a tuple of set variables, instead of just one variable.
A construct with the extended quantifier, $\U(X_1,\dots,X_k).\phi$, says that the subformula $\phi$ holds for tuples of sets in which each of $X_1,\dots,X_k$ is arbitrarily large.
This logic is capable of easily expressing properties in which multiple quantities are simultaneously required to be unbounded.
In particular, it can express the simultaneous unboundedness property, and thus it is strictly more expressive than the standard \wmsou logic:

\begin{theorem}\label{thm:expressivity}
	The \wmsouk logic can express some properties of trees that are not expressible in \wmsou;
	in particular, this is the case for the simultaneous unboundedness property.
\end{theorem}

In fact, to separate the two logics it is enough to consider \wmsouk only with \U quantifiers for pairs of variables (i.e., with $k=2$).
Actually, we are convinced that the proof of \cref{thm:expressivity} contained in this paper can be modified for showing that, for every $k\geq 2$,
\wmsouk without \U quantifiers for tuples of length at least $k$ is less expressive than \wmsouk with such quantifiers (cf.\@\cref{remark3}).

Our main theorem says that the model-checking procedure for \wmsou can be extended to the new logic:

\begin{theorem}\label{thm:decidability}
	Given an \wmsouk sentence $\phi$ and a recursion scheme $\Gg$ one can decide whether $\phi$ is true in the tree generated by $\Gg$.
\end{theorem}

\section{Preliminaries}

The powerset of a set $X$ is denoted $\Pp(X)$.
For $i,j\in\Nat$ we define $\interval{i,j}=\set{k\in\Nat\mid i\leq k\leq j}$.
The domain of a function $f$ is denoted $\dom(f)$.
When $f$ is a function, by $f[x\mapsto y]$ we mean the function that maps $x$ to $y$ and every other $z\in\dom(f)$ to $f(z)$.

\subparagraph{Trees.}

We consider rooted, potentially infinite trees, where children are ordered.
For simplicity of the presentation, we consider only binary trees, where every node has at most two children.
This is not really a restriction.
Indeed, it is easy to believe that our proofs can be generalized to trees of arbitrary bounded finite arity without any problem (except for notational complications).
Alternatively, a tree of arbitrary bounded finite arity can be converted into a binary tree using the first child / next sibling encoding,
and a logical formula can be translated as well to a formula talking about the encoding;
this means that the \wmsouk model-checking problem over trees of arbitrary bounded finite arity can be reduced to such a problem over binary trees.

Formally, a \emph{tree domain} (a set of tree nodes) is a set $D\subseteq\set{\L,\R}^*$ that is closed under taking prefixes (i.e., if $uv\in D$ then also $u\in D$).
A \emph{tree} over an alphabet $\AlA$ is a function $T\colon D\to\AlA$, for some tree domain $D$.
The set of trees over $\AlA$ is denoted $\Tt(\AlA)$.
The \emph{subtree} of $T$ starting in a node $v$ is denoted $T\restr v$ and is defined by $(T\restr v)(u)=T(vu)$ (with domain $\set{u\in\set{\L,\R}^*\mid vu\in\dom(T)}$).
For nodes we employ the usual notions of child, parent, ancestor, descendant, etc.\@ (where we assume that a node is also an ancestor and a descendant of itself).

For trees $T_1, T_2$, and for $a\in\AlA$ we write $a\symb{T_1,T_2}$ for the tree $T$ such that $T\restr\L=T_1$, $T\restr\R=T_2$, and $T(\epsilon)=a$.
We also write $\bot$ for the tree with empty domain.

\subparagraph{Recursion schemes.}

Recursion schemes are grammars used to describe some infinite trees in a finitary way.
We introduce recursion schemes only by giving an example, rather than by defining them formally.
This is enough, because this paper does not work with recursion schemes directly; it only uses some facts concerning them.

A recursion scheme is given by a set of rules, like this:
\begin{align*}
	&\S\to \F\,\G\,, &
	&\D\,\g\,\x\to \g\,(\g\,\x)\,,\\
	&\F\,\g\to\a\symb{\g\,\bot,\F\,(\D\,\g)}\,, &
	&\G\,\x\to\b\symb{\x,\bot}\,.
\end{align*}
Here $\S,\D,\F,\G$ are nonterminals, with $\S$ being the starting nonterminal, $\x,\g$ are variables, and $\a,\b$ are letters from $\AlA$.
To generate a tree, we start with $\S$, which reduces to $\F\,\G$ using the first rule.
We now use the rule for $\F$, where the parameter $\g$ is instantiated to be $\G$; we obtain $\a\symb{\G\,\bot,\F\,(\D\,\G)}$.
This already defines the root of the tree, which should be $\a$-labeled;
its two subtrees should be generated from $\G\,\bot$ and $\F\,(\D\,\G)$, respectively.
We see that $\G\,\bot$ reduces to $\b\symb{\bot,\bot}$, which is a tree with a single $\b$-labeled node.
On the other hand, $\F\,(\D\,\G)$ reduces to $\a\symb{\D\,\G\,\bot,\F\,(\D\,(\D\,\G))}$,
which means that the right child of the root is $\a$-labeled,
and its left subtree generated from $\D\,\G\,\bot$ (which reduces to $\G\,(\G\,\bot)$, then to $\b\symb{\G\bot,\bot}$, and then to $\b\symb{\b\symb{\bot,\bot},\bot}$)
is a path consisting of two $\b$-labeled nodes.
Continuing like this, when going right we always obtain a next $\a$-labeled node (we thus have an infinite $\a$-labeled branch),
and to the left of the $i$-th such node we have a tree generated from $\underbrace{\D\,(\D\,(\dots(\D}_{i-1}\G)\dots))\,\bot$,
which is a finite branch consisting of $2^{i-1}$ $\b$-labeled nodes (note that every $\D$ applies its argument twice, and hence doubles the number of produced $\b$-labeled nodes).
The resulting tree is depicted on \cref{fig:example-scheme}.
\begin{figure}
	\centering
	\def\svgscale{0.5}\import{pics/}{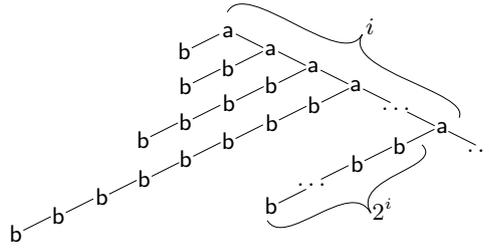}\vspace{-1ex}
	\caption{The tree generated by the example recursion scheme}
	\label{fig:example-scheme}
\end{figure}

For a formal definition of recursion schemes consult prior work (e.g., \cite{KNU-hopda,KobayashiACM,KrivineWS,wmsou-schemes}).
Some of these papers use a lambda-calculus notation, where our rule for $\D$ would be rather written as $\D\to\lambda\g.\lambda\x.\g\,(\g\,\x)$.
Sometimes it is also allowed to have $\lambda$ inside a rule, like $\S\to \F\,(\lambda\,\x.\b\symb{\x,\bot})$;
this does not make the definition more general, because subterms starting with $\lambda$ can be always extracted to separate nonterminals.

\section{The \texorpdfstring{\wmsouk}{WMSOU+tup} logic}

In this section we introduce the logic under consideration: the \wmsouk logic.

\subparagraph{Definition.}

For technical convenience, we use a syntax in which there are no first-order variables.
It is easy to translate a formula from a more standard syntax to ours:
first-order variables may be simulated by set variables for which we check that they contain exactly one node
(i.e., that they are nonempty and that every subset thereof is either empty or equal to the whole set).

We assume an infinite set $\Vv$ of variables, which can be used to quantify over finite sets of tree nodes.
In order to distinguish variables from sets to which these variables are valuated, we denote variables using Sans Serif font (e.g., $\X,\Y,\Z$).
In the syntax of \wmsouk we have the following constructions:
\begin{align*}
	\phi::= a(\X)\mid
		\X\child_d \Y\mid
		\X\subseteq \Y\mid
		\phi_1\land\phi_2\mid
		\neg\phi'\mid
		\efin\X.\phi'\mid
		\U(\X_1,\dots,\X_k).\phi'\,,
\end{align*}
where $a\in\AlA$, $d\in\set{\L,\R}$, $k\in\Nat$, and $\X,\Y,\X_1,\dots,\X_k\in\Vv$.
Free variables of a formula are defined as usual; in particular $\U(\X_1,\dots,\X_k)$ is a quantifier that bounds the variables $\X_1,\dots,\X_k$.

We evaluate formulae of \wmsouk in $\AlA$-labeled trees.
In order to evaluate a formula $\phi$ in a tree $T$, we also need a \emph{valuation}, that is, a function $\nu$ from $\Vv$ to finite sets of nodes of $T$
(its values are meaningful only for free variables of $\phi$).
The semantics of formulae is defined as follows:
\begin{itemize}
\item	$a(\X)$ holds when every node in $\nu(\X)$ is labeled with $a$,
\item	$\X\child_d \Y$ holds when both $\nu(\X)$ and $\nu(\Y)$ are singletons, and the unique node in $\nu(\Y)$ is the left (if $d=\L$) / right (if $d=\R$) child of the unique node in $\nu(\X)$,
\item	$\X\subseteq \Y$ holds when $\nu(\X)\subseteq\nu(\Y)$,
\item	$\phi_1\land\phi_2$ holds when both $\phi_1$ and $\phi_2$ hold,
\item	$\neg\phi'$ holds when $\phi'$ does not hold,
\item	$\efin\X.\phi'$ holds if there exists a finite set $X$ of nodes of $T$ for which $\phi'$ holds under the valuation $\nu[\X\mapsto X]$, and
\item	$\U(\X_1,\dots,\X_k).\phi'$ holds if for every $n\in\Nat$ there exist finite sets $X_1,\dots,X_k$ of nodes of $T$, each of cardinality at least $n$,
	such that $\phi'$ holds under the valuation $\nu[\X_1\mapsto X_1,\dots,\X_k\mapsto X_k]$.
\end{itemize}
We write $T,\nu\models\phi$ to denote that $\phi$ holds in $T$ under the valuation $\nu$.

\subparagraph{Logical types.}

In proofs of both our results, \cref{thm:expressivity} and \cref{thm:decidability}, we use logical types, which we now define.

Let $\phi$ be a formula of \wmsouk, let $T$ be a tree, and let $\nu$ be a valuation.
We define the \emph{$\phi$-type} of $T$ under valuation $\nu$, denoted $\pht{\phi}{T}{\nu}$, by induction on the size of $\phi$ as follows:
\begin{itemize}
\item	if $\phi$ is of the form $a(\X)$ (for some letter $a\in\AlA$) or $\X\subseteq \Y$ then $\pht{\phi}{T}{\nu}$ is the logical value of $\phi$ in $T,\nu$, that is,
	$\true$ if $T,\nu\models\phi$ and $\false$ otherwise,
\item	if $\phi$ is of the form $\X\child_d \Y$, then $\pht{\phi}{T}{\nu}$ equals:
	\begin{itemize}
	\item	$\true$ if $T,\nu\models\phi$,
	\item	$\mathsf{empty}$ if $\nu(\X)=\nu(\Y)=\emptyset$,
	\item	$\mathsf{root}$ if $\nu(\X)=\emptyset$ and $\nu(Y)=\set{\epsilon}$, and
	\item	$\false$ otherwise,
	\end{itemize}
\item	if $\phi=(\psi_1\land\psi_2)$, then $\pht{\phi}{T}{\nu}=(\pht{\psi_1}{T}{\nu},\pht{\psi_2}{T}{\nu})$,
\item	if $\phi=(\neg\psi)$, then $\pht{\phi}{T}{\nu}=\pht{\psi}{T}{\nu}$,
\item	if $\phi=\efin\X.\psi$, then
	\begin{align*}
		\pht{\phi}{T}{\nu}=\set{\sigma\mid\exists X.\,\pht{\psi}{T}{\nu[\X\mapsto X]}=\sigma}\,,
	\end{align*}
	where $X$ ranges over finite sets of nodes of $T$, and
\item	if $\phi=\U(\X_1,\dots,\X_k).\psi$, then
	\begin{align*}
		\pht{\phi}{T}{\nu}=\big(\set{\sigma\mid\forall n\in\Nat.\,\exists X_1\lamdots\exists X_k.\,&\pht{\psi}{T}{\nu[\X_1\mapsto X_1,\dots,\X_k\mapsto X_k]}=\sigma\\
			&\land\forall i\in I.\, |X_i|\geq n}\big)_{I\subseteq\interval{1,k}}\,,
	\end{align*}
	where $X_1,\dots,X_k$ range over finite sets of nodes of $T$ (the above $\phi$-type is a tuple of $2^k$ sets, indexed by subsets $I$ of $[1,k]$).
\end{itemize}

For each $\phi$, let $\Pht\phi$ denote the set of all potential $\phi$-types.
Namely, $\Pht\phi=\set{\true,\false}$ if $\phi=a(\X)$ or $\phi=(\X\subseteq\Y)$, $\Pht\phi=\set{\true, \mathsf{empty}, \mathsf{root}, \false}$ if $\phi=\X\child_d\Y$,
$\Pht\phi=\Pht{\psi_1}\times\Pht{\psi_2}$ if $\phi=(\psi_1\land\psi_2)$, $\Pht\phi=\Pht{\psi}$ if $\phi=(\neg\psi)$; $\Pht\phi=\Pp(\Pht\psi)$ if $\phi=\efin\X.\psi$, and
$\Pht\phi=(\Pp(\Pht\psi))^{2^k}$ if $\phi=\U(\X_1,\dots,\X_k).\psi$.

The following two \lcnamecrefs{prop:pht-finite} can be shown by a straightforward induction on the structure of a considered formula:

\begin{fact}\label{prop:pht-finite}
	For every \wmsouk formula $\phi$ the set $\Pht\phi$ is finite.
\end{fact}

The second \lcnamecref{prop:pht-2-form} says that whether or not $\phi$ holds in $T,\nu$ is determined by $\pht{\phi}{T}{\nu}$:

\begin{fact}\label{prop:pht-2-form}
	For every \wmsouk formula $\phi$ there is a computable function $\mathit{tv}_\phi\colon\Pht\phi\to\set{\true,\false}$ such that
	for every tree $T\in\Tt(\AlA)$ and every valuation $\nu$ in $T$, it holds that $\mathit{tv}_\phi(\pht{\phi}{T}{\nu})=\true$ if, and only if, $T,\nu\models\phi$.
\end{fact}

Next, we observe that types behave in a compositional way, as formalized below.
Here, for a node $w$ we write $X\restr w$ and $\nu\restr w$ to denote the restriction of a set $X$ and of a valuation $\nu$ to the subtree starting at $w$;
formally, $X\restr w=\set{u\mid wu\in X}$ and $\nu\restr w$ maps every variable $\X\in\Vv$ to $\nu(\X)\restr w$.

\begin{restatable}{proposition}{compositionality}\label{lem:compositionality}
	For every letter $a\in\AlA$ and every formula $\phi$, one can compute a function $\Comp_{a,\phi}\colon\Pp(\Vv)\times\Pht\phi\times\Pht\phi\to\Pht\phi$ such that
	for every tree $T$ whose root has label $a$ and for every valuation $\nu$,
	\begin{align}
		\pht{\phi}{T}{\nu}=\Comp_{a,\phi}(\set{\X\mid\epsilon\in\nu(\X)},\pht{\phi}{T\restr\L}{\nu\restr\L},\pht{\phi}{T\restr\R}{\nu\restr\R})\,.\label{eq:compos}
	\end{align}
\end{restatable}

We remark that \textit{a priori} the first argument of $\Comp_{a,\phi}$ is an arbitrary subset of $\Vv$, but in fact we only need to know which free variables of $\phi$ it contains;
in consequence, $\Comp_{a,\phi}$ can be seen as a finite object.

\begin{proof}[Proof of \cref{lem:compositionality}]
	We proceed by induction on the size of $\phi$.

	When $\phi$ is of the form $b(\X)$ or $\X\subseteq\Y$, then we see that $\phi$ holds in $T,\nu$ if, and only if,
	it holds in the subtrees $T\restr\L,\nu\restr\L$ and $T\restr\R,\nu\restr\R$, and in the root of $T$.
	Thus for $\phi=b(\X)$ as $\Comp_{a,\phi}(S,\tau_\L,\tau_\R)$ we take $\true$ when $\tau_\L=\tau_\R=\true$ and either $a=b$ or $\X\not\in S$.
	For $\phi=(\X\subseteq \Y)$ the last part of the condition is replaced by ``if $\X\in S$ then $\Y\in S$''.

	Next, suppose that $\phi=(\X\child_d\Y)$.
	Then as $\Comp_{a,\phi}(S,\tau_\L,\tau_\R)$ we take
	\begin{itemize}
	\item	$\true$ if $\X\not\in S$, $\Y\not\in S$, and either $\tau_\L=\true$ and $\tau_\R=\mathsf{empty}$ or $\tau_\L=\mathsf{empty}$ and $\tau_\R=\true$,
	\item	$\true$ also if $\X\in S$, $\Y\not\in S$, $\tau_d=\mathsf{root}$, and $\tau_i=\mathsf{empty}$ for the direction $i$ other than $d$,
	\item	$\mathsf{empty}$ if $\X\not\in S$, $\Y\not\in S$, and $\tau_L=\tau_S=\mathsf{empty}$,
	\item	$\mathsf{root}$ if $\X\not\in S$, $\Y\in S$, and $\tau_L=\tau_S=\mathsf{empty}$, and
	\item	$\false$ otherwise.
	\end{itemize}
	By comparing this definition with the definition of the type we immediately see that \cref{eq:compos} is satisfied.

	When $\phi=(\neg\psi)$, we simply take $\Comp_{a,\phi}=\Comp_{a,\psi}$,
	and when $\phi=(\psi_1\land\psi_2)$, as $\Comp_{a,\phi}(S,(\tau_\L^1,\tau_\L^2),(\tau^1_\R,\tau^2_\R))$ we take the pair of $\Comp_{a,\psi_i}(S,\tau_\L^i,\tau^i_\R)$ for $i\in\set{1,2}$.

	Suppose now that $\phi=\efin\X.\psi$.
	We define $\Comp_{a,\phi}(S,\tau_\L,\tau_\R)$ to be
	\begin{align*}
		\set{\Comp_{a,\psi}(S',\sigma_\L,\sigma_\R)\mid S\setminus\set{\X}\subseteq S'\subseteq S\cup\set{\X},\sigma_\L\in\tau_\L,\sigma_\R\in\tau_\R}\,.
	\end{align*}
	Let us check \cref{eq:compos} in details.
	Denote $S=\set{\Y\mid\epsilon\in\nu(\Y)}$.
	In order to show the left-to-right inclusion recall that, by definition,
	$\pht{\phi}{T}{\nu}$ is a set of $\psi$-types, whose every element is of the form $\pht{\psi}{T}{\nu[\X\mapsto X]}$ for some finite set of nodes $X$.
	For every such $X$ by the induction hypothesis we have
	$\pht{\psi}{T}{\nu[\X\mapsto X]}=\Comp_{a,\psi}(S',\pht{\phi}{T\restr\L}{\nu[\X\mapsto X]\restr\L},\pht{\phi}{T\restr\R}{\nu[\X\mapsto X]\restr\R})$,
	where $S'=S\cup\set{\X}$ if $\epsilon\in X$ and $S'=S\setminus\set{\X}$ if $\epsilon\not\in X$;
	moreover $\pht{\psi}{T\restr\L}{\nu[\X\mapsto X]\restr\L}\in\pht{\phi}{T\restr L}{\nu\restr\L}$
	and $\pht{\psi}{T\restr\R}{\nu[\X\mapsto X]\restr\R}\in\pht{\phi}{T\restr L}{\nu\restr\R}$,
	which implies that $\pht{\psi}{T}{\nu[\X\mapsto X]}\in\Comp_{a,\phi}(S,\pht{\phi}{T\restr\L}{\nu\restr\L},\pht{\phi}{T\restr\R}{\nu\restr\R})$, as required.
	For the opposite inclusion take some $\sigma\in\Comp_{a,\phi}(S,\pht{\phi}{T\restr\L}{\nu\restr\L},\pht{\phi}{T\restr\R}{\nu\restr\R})$;
	it is of the form $\Comp_{a,\psi}(S',\sigma_\L,\sigma_\R)$ for some $\sigma_\L\in\pht{\phi}{T\restr\L}{\nu\restr\L}$ and $\sigma_\R\in\pht{\phi}{T\restr\R}{\nu\restr\R}$,
	where $S'$ is either $S\cup\set{\X}$ or $S\setminus\set{\X}$.
	Then, by definition, $\sigma_\L$ and $\sigma_\R$ are of the form $\pht{\psi}{T\restr\L}{(\nu\restr\L)[\X\mapsto X_\L]}$
	and $\pht{\psi}{T\restr\R}{(\nu\restr\R)[\X\mapsto X_\R]}$, respectively, for some finite sets of nodes $X_\L$ and $X_\R$.
	We now take $X$ such that $X\restr\L=X_\L$ and $X\restr\R=X_\R$, and $\epsilon\in X$ if, and only if, $S'=S\cup\set{\X}$;
	we have $(\nu\restr\L)[\X\mapsto X_\L]=\nu[\X\mapsto X]\restr\L$ and $(\nu\restr\R)[\X\mapsto X_\R]=\nu[\X\mapsto X]\restr\R$.
	By the induction hypothesis we then have $\sigma=\pht{\psi}{T}{\nu[\X\mapsto X]}$, which by definition is an element of $\pht{\phi}{T}{\nu}$, as required.

	Finally, suppose that $\phi=\U(\X_1,\dots,\X_k).\psi$.
	For $\tau_\L=(\rho_{\L,I})_{I\subseteq\interval{1,k}}$ and $\tau_\R=(\rho_{\R,I})_{I\subseteq\interval{1,k}}$
	we define $\Comp_{a,\phi}(S,\tau_\L,\tau_\R)$ to be $(\rho_I)_{I\subseteq\interval{1,k}}$, where
	\begin{align*}
		\rho_I=\set{\Comp_{a,\psi}(S',\sigma_\L,\sigma_\R)\mid{} &S\setminus\set{\X_1,\dots,\X_k}\subseteq S'\subseteq S\cup\set{\X_1,\dots,\X_k},\\
			&\sigma_\L\in\rho_{\L,I_\L},\sigma_\R\in\rho_{\R,I_\R},I_\L\cup I_\R=I}\,.
	\end{align*}
	In order to check \cref{eq:compos}, denote $\pht{\phi}{T}{\nu}=(\rho_I')_{I\subseteq\interval{1,k}}$,
	$\pht{\phi}{T\restr\L}{\nu\restr\L}=(\rho_{\L,I})_{I\subseteq\interval{1,k}}$,
	$\pht{\phi}{T\restr\R}{\nu\restr\R}=(\rho_{\R,I})_{I\subseteq\interval{1,k}}$, and $S=\set{\Y\mid\epsilon\in\nu(\Y)}$;
	we then have to prove that $\rho'_I=\rho_I$ for all $I\subseteq\interval{1,k}$ (where $\rho_I$ is as defined above).

	For the left-to-right inclusion, take some $\sigma\in\rho'_I$.
	By definition, it is a $\psi$-type such that for every $n\in\Nat$ there exist finite sets $X_{n,1},\dots,X_{n,k}$ for which
	$\pht{\psi}{T}{\nu[\X_1\mapsto X_{n,1},\dots,\X_k\mapsto X_{n,k}]}=\sigma$, where the cardinality of the sets $X_{n,i}$ with $i\in I$ is at least $n$.
	To every $n$ let us assign the following information, called \emph{characteristic}, and consisting of $2k$ bits and $2$ $\psi$-types:
	\begin{itemize}
	\item	for every $i\in\interval{1,k}$, does the root $\epsilon$ belong to $X_{n,i}$?
	\item	for every $i\in\interval{1,k}$, is $X_{n,i}\restr\L$ larger than $X_{n,i}\restr\R$?
	\item	the $\psi$-types $\pht{\psi}{T\restr\L}{\nu[\X_1\mapsto X_{n,1},\dots,\X_k\mapsto X_{n,k}]\restr\L}$ and
		$\pht{\psi}{T\restr\R}{\nu[\X_1\mapsto X_{n,1},\dots,\X_k\mapsto X_{n,k}]\restr\R}$.
	\end{itemize}
	Because there are only finitely many possible characteristics,
	by the pigeonhole principle we may find an infinite set $G\subseteq\Nat$ of indices $n$ such that the same characteristic is assigned to every $n\in G$.
	We then take
	\begin{align*}
		&S'=S\setminus\set{\X_1,\dots,\X_k}\cup\set{\X_i\mid\epsilon\in X_{n,i}\mbox{ for }n\in G}\,,\\
		&I_\L=\set{i\in I\mid |X_{n,i}\restr\L|>|X_{n,i}\restr\R|\mbox{ for }n\in G}\,,\qquad
			I_\R=I\setminus I_\L\,,\\
		&\sigma_\L=\pht{\psi}{T\restr\L}{\nu[\X_1\mapsto X_{n,1},\dots,\X_k\mapsto X_{n,k}]\restr\L}\,,\qquad\!
		\sigma_\R=\pht{\psi}{T\restr\R}{\nu[\X_1\mapsto X_{n,1},\dots,\X_k\mapsto X_{n,k}]\restr\R}\qquad\!\mbox{for $n\in G$.}
	\end{align*}
	The induction hypothesis (used with the valuation $\nu[\X_1\mapsto X_{n,1},\dots,\X_k\mapsto X_{n,k}]$ for any $n\in G$)
	gives us $\sigma=\Comp_{a,\psi}(S',\sigma_\L,\sigma_\R)$.
	For every $m\in\Nat$ we can find $n\in G$ such that $n\geq 2m+1$;
	then $\pht{\psi}{T\restr\L}{\nu[\X_1\mapsto X_{n,1},\dots,\X_k\mapsto X_{n,k}]\restr\L}=\sigma_\L$ and $|X_{n,i}\restr\L|\geq m$ for all $i\in I_\L$
	($X_{n,i}$ has at least $2m+1$ elements because $i\in I$, one of them may be the root, and at least half of the other elements is in the left subtree by definition of $I_\L$).
	This implies that $\sigma_\L\in\rho_{\L,I}$, by definition of the $\phi$-type.
	Likewise $\sigma_\R\in\rho_{\R,I}$.
	By definition of $\rho_I$ this gives us $\sigma\in\rho_I$ as required.

	The right-to-left inclusion is completely straightforward.
	Indeed, take some $\sigma\in\rho_I$.
	The definition of $\rho_I$ gives us a set $S'$ such that $S\setminus\set{\X_1,\dots,\X_k}\subseteq S'\subseteq S\cup\set{\X_1,\dots,\X_k}$,
	sets $I_\L,I_\R\subseteq\interval{1,k}$ such that $I_\L\cup I_\R=I$, and types $\sigma_\L\in\rho_{\L,I_\L}$ and $\sigma_\R\in\rho_{\R,I_\R}$
	such that $\sigma=\Comp_{a,\psi}(S',\sigma_\L,\sigma_\R)$.
	By definition of the two $\psi$-types, $\sigma_\L$ and $\sigma_\R$, for every $n$ there are sets $X_{\L,1},\dots,X_{\L,k}$ and $X_{\R,1},\dots,X_{\R,k}$ such that
	$\pht{\psi}{T\restr\L}{(\nu\restr\L)[\X_1\mapsto X_{\L,1},\dots,\X_k\mapsto X_{\L,k}]}=\sigma_\L$,
	and $\pht{\psi}{T\restr\R}{(\nu\restr\R)[\X_1\mapsto X_{\R,1},\dots,\X_k\mapsto X_{\R,k}]}=\sigma_\R$,
	and $|X_{\L,i}|\geq n$ for all $i\in I_\L$, and $|X_{\R,i}|\geq n$ for all $i\in I_\R$.
	We now take $X_1,\dots,X_k$ such that $X_i\restr\L=X_{\L,i}$, and $X_i\restr\R=X_{\R,i}$, and $\epsilon\in X_i$ if, and only if, $\X_i\in S'$, for all $i\in\interval{1,k}$.
	By the induction hypothesis we then have $\pht{\psi}{T}{\nu[\X_1\mapsto X_1,\dots,\X_k\mapsto X_k]}=\Comp_{a,\psi}(S',\sigma_\L,\sigma_\R)=\sigma$.
	Because additionally $|X_i|\geq n$ for all $i\in I=I_\L\cup I_\R$, we obtain that $\sigma\in\rho'_I$, as required.
\end{proof}

The next \lcnamecref{lem:type-empty} says that one can find a type of the empty tree.
In the empty tree, a valuation has to map every variable to the empty set; we denote such a valuation by $\varnothing$.
This \lcnamecref{lem:type-empty} is trivial: we simply follow the definition of $\pht{\phi}{\bot}{\varnothing}$.

\begin{fact}\label{lem:type-empty}
	For ever formula $\phi$, one can compute $\pht{\phi}{\bot}{\varnothing}$.
\end{fact}

\section{Decidability of model-checking}

In this section we show how to evaluate \wmsouk formulae over trees generated by recursion schemes, that is, we prove \cref{thm:decidability}.
To this end, we first introduce three kinds of operations on recursion schemes, known to be computable.
Then, we show how a sequence of these operations can be used to evaluate our formulae.

\subparagraph{\mso reflection.}

The property of logical reflection for recursion schemes comes from Broadbent, Carayol, Ong, and Serre~\cite{reflection-selection}.
They state it for sentences of $\mu$-calculus,
but $\mu$-calculus and \mso are equivalent over infinite trees~\cite{EmersonJutla}.

Consider a tree $T$, and an \mso sentence $\phi$ (we skip a formal definition of \mso, assuming that it is standard).
We define $\refl_\phi(T)$ to be the tree having the same domain as $T$, and such that every node $u$ thereof is labeled by the pair $(a_u,b_u)$,
where $a_u$ is the label of $u$ in $T$, and $b_u$ is $\true$ if $\phi$ is satisfied in $T\restr u$ and $\false$ otherwise.
In other words, $\refl_\phi(T)$ adds, in every node of $T$, a mark saying whether $\phi$ holds in the subtree starting in that node.
Consult \cref{fig:mso-refl} for an example.
\begin{figure}
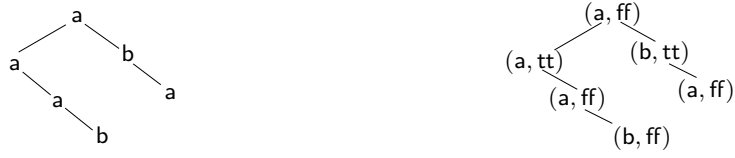

	\centering
	\begin{minipage}[c]{0.49\textwidth}
		\centering
		\def\svgscale{0.5}\import{pics/}{mso-refl-1.pdf_tex_ok}\vspace{-1ex}
	\end{minipage}
	\begin{minipage}[c]{0.49\textwidth}
		\centering
		\def\svgscale{0.5}\import{pics/}{mso-refl-2.pdf_tex_ok}\vspace{-1ex}
	\end{minipage}
	\caption{An example tree $T$ (left), and the corresponding tree $\refl_\phi(T)$ (right) obtained for an MSO sentence $\phi$ saying ``the right child of the root has label $\a$''}
	\label{fig:mso-refl}
\end{figure}

\begin{theorem}[\mso reflection {\cite[Theorem 7.3(2)]{reflection-selection}}]\label{fact:reflection}
	Given a recursion scheme $\Gg$ generating a tree $T$, and an \mso sentence $\phi$, one can construct a recursion scheme $\Gg_\phi$
	generating the tree $\refl_\phi(T)$.
\end{theorem}

\subparagraph{SUP reflection.}

The SUP reflection is the heart of our proof (where ``SUP'' stands for simultaneous unboundedness property).
In order to talk about this property, we need a few more definitions.
By $\#_a(V)$ we denote the number of $a$-labeled nodes in a (finite) tree $V$.
For a set of (finite) trees $\Ll$ and a set of letters $A$, we define a predicate $\SUP_A(\Ll)$,
which holds if for every $n\in\Nat$ there is some $V_n\in \Ll$ such that for all $a\in A$ it holds that $\#_a(V_n)\geq n$.

Originally, in the simultaneous unboundedness property we consider devices recognizing a set of finite trees, unlike recursion schemes, which generate a single infinite tree.
We use here an equivalent formulation, in which the set of finite trees is encoded in a single infinite tree.
To this end, we use two special letters: $\nd$, denoting a nondeterministic choice (disjunction between two children),
and $\nd_\bot$, denoting that there is no choice (empty disjunction).
We write $T\rednd V$ if $V$ is obtained from $T$ by choosing some $\nd$-labeled node $u$ and some its child $v$, and attaching $T\restr v$ in place of $T\restr u$.
In other words, $\rednd$ is the smallest relation such that $\nd\symb{T_\L,T_\R}\rednd T_d$ for $d\in\set{\L,\R}$,
and $a\symb{T_\L,T_\R}\rednd a\symb{T_\L',T_\R}$ if $T_\L\rednd T_\L'$,
and $a\symb{T_\L,T_\R}\rednd a\symb{T_\L,T_\R'}$ if $T_\R\rednd T_\R'$.
For a tree $T$, $\Ll(T)$ is the set of all finite trees $V$ such that $\#_\nd(V)=\#_{\nd_\bot}(V)=0$ and $T\rednd^* V$.
See \cref{fig:example-nd} for an example.
\begin{figure}
	\centering
	\begin{minipage}[c]{0.49\textwidth}
		\centering
		\def\svgscale{0.5}\import{pics/}{nd-1.pdf_tex_ok}\vspace{-1ex}
	\end{minipage}
	\begin{minipage}[c]{0.24\textwidth}
		\centering
		\def\svgscale{0.5}\import{pics/}{nd-2.pdf_tex_ok}\vspace{-1ex}\\[2ex]
		\def\svgscale{0.5}\import{pics/}{nd-4.pdf_tex_ok}\vspace{-1ex}
	\end{minipage}
	\begin{minipage}[c]{0.24\textwidth}
		\def\svgscale{0.5}\import{pics/}{nd-3.pdf_tex_ok}\vspace{-1ex}\\[2ex]
		\def\svgscale{0.5}\import{pics/}{nd-5.pdf_tex_ok}\vspace{-1ex}
	\end{minipage}
	\caption{An example tree $T$, and four trees in $\Ll(T)$ (right).
		Additionally, $\Ll(T)$ contains $\bot$, the tree with empty domain, obtained by choosing the right child in the topmost $\nd$-labeled node.
		Note that no tree in $\Ll(T)$ contains $\d$, because $\d$ in $T$ is followed by $\nd_\bot$, which is forbidden in trees in $\Ll(T)$.}
	\label{fig:example-nd}
\end{figure}
We then say that $T$ satisfies the \emph{simultaneous unboundedness property with respect to} a set of letters $A$ if $\SUP_A(\Ll(T))$ holds,
that is, if for every $n\in\Nat$ there are trees in $\Ll(T)$ having at least $n$ occurrences of every letter from $A$.

Let $T$ be a tree over an alphabet $\AlA$.
We define $\refl_\SUP(T)$ to be the tree having the same domain as $T$,
and such that every node $u$ thereof, having in $T$ label $a_u$, is labeled by
\begin{itemize}
\item	the pair $(a_u,\{A\subseteq\AlA\mid \SUP_A(\Ll(T\restr u))\})$, if $a_u\not\in\set{\nd,\nd_\bot}$, and
\item	the original letter $a_u$, if $a_u\in\set{\nd,\nd_\bot}$.
\end{itemize}
In other words, $\refl_\SUP(T)$ adds, in every node $u$ of $T$ (except for $\nd$- and $\nd_\bot$-labeled nodes) and for every set $A$ of letters,
a mark saying whether $T\restr u$ has the simultaneous unboundedness property with respect to $A$.

Consider, for example, the tree $T$ from \cref{fig:example-sup}.
The tree $\refl_\SUP(T)$ has the same shape as $T$.
Every node $u$ having label $\a$ in $T$ gets label $(\a,\set{\emptyset,\set{\a},\set{\b},\set{\c},\set{\a,\b},\set{\a,\c}})$.
Note that the set does not contain $\set{\b,\c}$ nor $\set{\a,\b,\c}$: in $\Ll(T\restr u)$ there are no trees having simultaneously many occurrences of $\b$ and many occurrences of $\c$.
Nodes $u$ having in $T$ label $\b$ or $\c$ are simply relabeled to $(\b,\set{\emptyset})$ or $(\c,\set{\emptyset})$, respectively,
because $\Ll(T\restr u)$ contains only a single tree, with a fixed number of nodes.
\begin{figure}
	\centering
	\def\svgscale{0.5}\import{pics/}{sup-1.pdf_tex_ok}\vspace{-1ex}
	\caption{A tree $T$ illustrating SUP reflection}
	\label{fig:example-sup}
\end{figure}

\begin{theorem}[SUP reflection {\cite[Theorem 10.1]{types-diagonal-journal}}]\label{fact:sup-reflection}
	Given a recursion scheme $\Gg$ generating a tree $T$, one can construct a recursion scheme $\Gg_\SUP$ generating the tree $\refl_\SUP(T)$.
\end{theorem}

\begin{remark}\label{remark1}
	In the introduction we have described an easier variant of the simultaneous unboundedness property,
	called a \emph{word variant}.
	In this variant, every node with label other than $\nd$ has at most one child;
	then choosing a tree in $\Ll(T)$ corresponds to choosing a branch of $T$
	(and trees in $\Ll(T)$ consist of single branches, hence they can be seen as words).
	Although the word variant of SUP is more commonly known than the \emph{tree variant} described in this section,
	\cref{fact:sup-reflection} holds also for the more general tree variant, as presented above.
\end{remark}

\subparagraph{Transducers.}

A \emph{(deterministic, top-down) finite tree transducer} is a tuple $\Ff=(\AlA,\allowbreak\AlB,\allowbreak Q,\allowbreak q_0,\allowbreak\delta)$, where
$\AlA$ is a finite input alphabet, $\AlB$ is a finite output alphabet, $Q$ is a finite set of states, $q_0\in Q$ is an initial state,
and $\delta$ is a transition function mapping $Q\times(\AlA\cup\set{\bot})$ to finite trees over the alphabet $\AlB\cup(Q\times\set{\L,\R})$.
Letters from $Q\times\set{\L,\R}$ are allowed to occur only in leaves of trees $\delta(q,a)$ with $a\in\AlA$
(internal nodes of these trees, and all nodes of trees $\delta(q,\bot)$ are labeled by letters from $\AlB$).
Moreover, it is assumed that that there is no sequence $(q_1,a_1,d_1),(q_2,a_2,d_2),\dots,(q_n,a_n,d_n)$
such that $\delta(q_i,a_i)=(q_{(i\bmod n)+1},d_i)\symb{\bot,\bot}$ for all $i\in\interval{1,n}$.

For an input tree $T$ over $\AlA$ and a state $q\in Q$, we define an output tree $\Ff_q(T)$ over $\AlB$.
Namely $\Ff_q(a\symb{T_\L,T_\R})$ is the tree obtained from $\delta(q,a)$ by substituting $\Ff_r(T_d)$ for every leaf labeled with $(r,d)\in Q\times\set{\L,\R}$;
additionally, $\Ff_q(\bot)$ simply equals $\delta(q,\bot)$ (recall that this tree has no labels from $Q\times\set{\L,\R}$).
In other words, while being in state $q$ over an $a$-labeled node of the input tree, the transducer produces a tree prefix specified by $\delta(q,a)$,
where instead of outputting an $(r,\L)$-labeled (or $(r,\R)$-labeled) leaf, it rather continues by going to the left (respectively, right) child in the input tree, in state $r$;
when $\Ff$ leaves the domain of the input tree, it still has a chance to output something, namely $\delta(q,\bot)$, and then it stops.
In the root we start from the initial state, that is, we define $\Ff(T)=\Ff_{q_0}(T)$.
To make the above definition formal, we can define $\Ff_q(T)(v)$, the label of $\Ff_q(T)$ in a node $v\in\set{\L,\R}^k$, by induction on the depth $k$,
simultaneously for all input trees $T$ and states $q\in Q$.
Transitions $\delta(q,a)$ with $(r,d)$ immediately in the root are a bit problematic,
because we go down along the input tree without producing anything in the output tree;
we have assumed, however, that such transitions do not form a cycle, so after a few (at most $|Q|$) steps we necessarily advance in the output tree.

Note that transducers need not be linear.
For example, we may have $\delta(q,a)=a\symb{a\symb{(q,\L),\allowbreak(q,\L)},\allowbreak a\symb{(q,\R),(q,\R)}}$,
which creates two copies of the tree produced out of the left subtree, and two copies of the tree produced out of the right subtree.

We have the following \lcnamecref{fact:transducer}:

\begin{theorem}\label{fact:transducer}
	Given a finite tree transducer $\Ff=(\AlA,\AlB,Q,q_0,\delta)$ and a recursion scheme $\Gg$ generating a tree $T$ over the alphabet $\AlA$,
	one can construct a recursion scheme $\Gg_\Ff$ generating the tree $\Ff(T)$.
\end{theorem}

This \lcnamecref{fact:transducer} follows from the equivalence between recursion schemes and collapsible pushdown systems~\cite{collapsible}, as it is straightforward to compose a collapsible pushdown system with $\Ff$.
A formal proof can be found for instance in Parys~\cite[Appendix A]{types-diagonal-journal}.

\subparagraph{Sequences of operations.}

We consider sequences of operations of the form $O_1,O_2,\dots,O_n$, where every $O_i$ is either an \mso sentence $\phi$, or the string ``$\SUP$'', or a finite tree transducer $\Ff$.
Having a tree $T$, we can apply such a sequence of operations to it.
Namely, we take $T_0=T$, and for every $i\in\interval{1,n}$, as $T_i$ we take
\begin{itemize}
\item	$\refl_\phi(T_{i-1})$ if $O_i=\phi$ is an \mso sentence,
\item	$\refl_\SUP(T_{i-1})$ if $O_i=\SUP$, and
\item	$\Ff(T_{i-1})$ if $O_i=\Ff$ is a finite tree transducer.
\end{itemize}
As the result we take $T_n$.
We implicitly assume that whenever we apply a finite tree transducer to some tree, then the tree is over the input alphabet of the transducer;
likewise, we assume that while computing $\refl_\phi(T_{i-1})$, the formula uses letters from the alphabet of $T_{i-1}$.

Using the aforementioned closure properties (\cref{fact:reflection,fact:sup-reflection,fact:transducer}) we can apply the operations on the level of recursion schemes generating our tree:

\begin{proposition}\label{lem:sequence}
	Given a recursion scheme $\Gg$ generating a tree $T$, and a sequence of operations $O_1,O_2,\dots,O_n$ as above,
	one can construct a recursion scheme $\Gg'$ generating the result of applying $O_1,O_2,\dots,O_n$ to $T$.
\end{proposition}

\subparagraph{Main theorem.}

Let $\AlA$ be the alphabet used by \wmsouk formulae under consideration. We prove the following theorem:

\begin{theorem}\label{thm:create-sequence}
	Given a \wmsouk sentence $\phi$, one can compute a sequence of operations $O_1,O_2,\dots,O_n$,
	such that for every tree $T$ over $\AlA$,
	by applying $O_1,O_2,\dots,O_n$ to $T$ we obtain $\true\symb{\bot,\bot}$ if $\phi$ is true in $T$, and $\false\symb{\bot,\bot}$ otherwise.
\end{theorem}

Having a recursion scheme generating either $\true\symb{\bot,\bot}$ or $\false\symb{\bot,\bot}$, we can easily check what is generated:
we just repeatedly apply rules of the recursion scheme.
Thus \cref{thm:decidability} is an immediate consequence of \cref{thm:create-sequence,lem:sequence}.

\begin{remark}
	Note that in \cref{thm:create-sequence} we do not assume that $T$ is generated by a recursion scheme;
	the theorem holds for any tree $T$.
	Thus our decidability result, \cref{thm:decidability}, can be immediately generalized from the class of trees generated by recursion schemes
	to any class of trees that is effectively closed under the considered three types of operations
	(i.e., any class for which \cref{fact:reflection,fact:sup-reflection,fact:transducer} remain true).
\end{remark}

We now formulate a variant of \cref{thm:create-sequence} suitable for induction.
On the input side, we have to deal with formulae with free variables (subformulae of our original sentence).
On the output side, it is not enough to produce the truth value; we rather need to produce trees decorated by logical types.
While logical types in general depend on the valuation of free variables, we consider here only a very special valuation mapping all variables to the empty set;
recall that we denote such a valuation by $\varnothing$.
Additionally, in the input tree we have to allow presence of some additional labels (used to store types with respect to other subformulae):
we suppose that we have a tree $T$ over an alphabet $\AlA\times\AlB$, where $\AlA$ is our fixed alphabet used by \wmsouk formulae,
and $\AlB$ is some other auxiliary alphabet.
Then by $\pi_\AlA(T)$ we denote the tree over $\AlA$ having the same domain as $T$, with every node thereof relabeled from $(a,b)\in\AlA\times\AlB$ to $a$.

\begin{lemma}\label{lem:create-sequence}
	Given a \wmsouk formula $\phi$ and an auxiliary alphabet $\AlB$, one can compute a sequence of operations $O_1,O_2,\dots,O_n$,
	such that for every tree $T$ over $\AlA\times\AlB$,
	by applying $O_1,O_2,\dots,O_n$ to $T$ we obtain a tree having the same domain as $T$,
	such that every node $u$ thereof is labeled by the pair $(\ell_u,\pht{\phi}{\pi_\AlA(T)\restr u}{\varnothing})$,
	where $\ell_u\in\AlA\times\AlB$ is the label of $u$ in $T$.
\end{lemma}

\cref{thm:create-sequence} is an immediate consequence of \cref{lem:create-sequence}.
Indeed, let us use \cref{lem:create-sequence} with a singleton alphabet $\AlB$;
for such an alphabet we identify $\AlA$ with $\AlA\times\AlB$.
By applying operations $O_1,\dots,O_n$ obtained from \cref{lem:create-sequence} we obtain a tree with the root labeled by $(a,\tau)$ for $\tau=\pht{\phi}{T}{\varnothing}$.
Recall that, by \cref{prop:pht-2-form}, we have a function $\mathit{tv}_\phi$ such that $\mathit{tv}_\phi(\pht{\phi}{T}{\varnothing})=\true$ if, and only if, $T,\varnothing\models\phi$.
Thus, after all the operations $O_1,\dots,O_n$, we can simply apply a transducer $\Ff$ that reads the root's label $(a,\tau)$
and returns the tree $\true\symb{\bot,\bot}$ if $\mathit{tv}_\phi(\tau)=\true$, and the tree $\false\symb{\bot,\bot}$ otherwise.
There is a small exception if the original tree $T$ has empty domain: then there is no root at all, in particular no root from which we can read the $\phi$-type $\tau$.
Thus, if the transducer $\Ff$ sees an empty tree, it should rather use $\tau=\pht{\phi}{\bot}{\varnothing}$, which is known by \cref{lem:type-empty}.

\begin{proof}[Proof of \cref{lem:create-sequence}]
	The proof is by induction on the structure of $\phi$.
	We have several cases depending on the shape of $\phi$.

	Recall that in this lemma we only consider the valuation $\varnothing$ mapping all variables to the empty set.
	Because of that, if $\phi$ is of the form $a(X)$ or $X\subseteq Y$, then the $\phi$-type $\pht{\phi}{\pi_\AlA(T)\restr u}{\varnothing}$ is $\true$ for every tree $T$ and node $u$ thereof.
	It is thus enough to return (as the only operation $O_1$) a transducer that appends $\true$ to the label of every node of $T$.
	Similarly, if $\phi=(X\child_d Y)$, then the $\phi$-type $\pht{\phi}{\pi_\AlA(T)\restr u}{\varnothing}$ is always $\mathsf{empty}$.
	For $\phi=(\neg\psi)$ the situation is also trivial: we can directly use the induction hypothesis since
	$\pht{\phi}{\pi_\AlA(T)\restr u}{\varnothing}=\pht{\psi}{\pi_\AlA(T)\restr u}{\varnothing}$.

	Suppose that $\phi=(\psi_1\land\psi_2)$.
	The induction hypothesis for $\psi_1$ gives us a sequence of operations $O_1,O_2,\dots,O_n$ that appends $\pht{\psi_1}{\pi_\AlA(T)\restr u}{\varnothing}$
	to the label of every node $u$ of $T$.
	The resulting tree $T'$ is over the alphabet $\AlA\times\AlB\times\Pht{\psi_1}$, which can be seen as $\AlA\times\AlB'$ for $\AlB'=\AlB\times\Pht{\psi_1}$;
	we have $\pi_\AlA(T')=\pi_\AlA(T)$.
	We can thus apply the induction hypothesis for $\psi_2$ to the resulting tree $T'$;
	it gives us a sequence of operations $O_{n+1},O_{n+2},\dots,O_{n+m}$ that appends $\pht{\psi_2}{\pi_\AlA(T)\restr u}{\varnothing}$
	to the label of every node $u$ of $T'$.
	The tree obtained after applying all the $n+m$ operations is as needed: in every node thereof we have appended the pair
	containing the $\psi_1$-type and the $\psi_2$-type, and such a pair is precisely the $\phi$-type.

	The case of $\phi=\efin\X.\psi$ is handled by a reduction to the case of $\phi'=\U\X.\psi$.
	Indeed, recall that the type for $\U(\X_1,\dots,\X_k)$ is a tuple of $2^k$ coordinates indexed by sets $I\subseteq\interval{1,k}$;
	in the case of a single variable $\X_1=\X$, there are only two coordinates, one for $I=\emptyset$, and the other for $I=\set{1}$.
	The coordinate for $I=\emptyset$ in $\pht{\U\X.\psi}{T'}{\varnothing}$ is simply $\set{\sigma\mid\exists X.\pht{\psi}{T'}{\nu[\X\mapsto X]}=\sigma}$,
	that is, the $\phi$-type $\pht{\efin\X.\psi}{T'}{\nu}$.
	Thus, we can take the sequence of operations $O_1,O_2,\dots,O_n$ from the forthcoming case of $\phi'=\U\X.\psi$, which appends the $\phi'$-type,
	and then add a simple transducer that removes the second coordinate of this type.

	Finally, suppose that $\phi=\U(\X_1,\dots,\X_k).\psi$.
	By the induction hypothesis we have a sequence of operations $O_1,O_2,\dots,O_n$ that appends the $\psi$-type $\pht{\psi}{\pi_\AlA(T)\restr u}{\varnothing}$
	to the label of every node $u$ of $T$.
	Let $T^1$ be the tree obtained from $T$ by applying these operations.

	As a first step, to $T^1$ we apply a transducer $\Ff$ defined as follows.
	Its input alphabet is $\AlA'=\AlA\times\AlB\times\Pht{\psi}$, the alphabet of $T^1$,
	its output alphabet is $\AlA'\cup\set{?,\#,\nd,\nd_\bot,\X_1,\dots,\X_k}$,
	and its set of states is $\set{q_0}\cup\Pht{\psi}$.
	Having a letter $\ell=(a,b,\tau)\in\AlA'$, let $\pi_\AlA(\ell)=a$ and $\pi_{\Pht{\psi}}(\ell)=\tau$.
	Coming to transitions, first for every triple $(S,\tau_\L,\tau_\R)$, where $S=\set{\X_{i_1},\dots,\X_{i_m}}\subseteq\set{\X_1,\dots,\X_k}$ and $\tau_\L,\tau_\R\in\Pht{\psi}$
	we define
	\begin{align*}
		\mathit{sub}(S,\tau_\L,\tau_\R)=\X_{i_1}\symb{\bot,\X_{i_2}\symb{\bot,\dots\X_{i_m}\symb{\bot,\#\symb{(\tau_\L,\L),(\tau_\R,\R)}}\dots}}\,.
	\end{align*}
	Moreover, for every $\ell\in\AlA'$ and $\tau\in\Pht\psi$, let $\mathit{here}(\ell,\tau)=\bot$ if $\tau=\pi_{\Pht{\psi}}(\ell)$ and $\mathit{here}(\ell,\tau)=\nd_\bot\symb{\bot,\bot}$ otherwise.
	In order to define $\delta(\tau,\ell)$, we consider all triples $(S_1,\tau_{\L,1},\tau_{\R,1}),\dots,\allowbreak (S_s,\tau_{\L,s},\tau_{\R,s})$
	for which $\Comp_{\pi_{\AlA}(\ell),\psi}(S_i,\tau_{\L,i},\tau_{\R,i})=\tau$ (assuming some fixed order in which these triples are listed).
	Then, we take
	\begin{align*}
		\delta(\tau,\ell)={}?\symb{\bot,\nd\symb{\mathit{sub}(S_1,\tau_{\L,1},\tau_{\R,1}),\nd\symb{\mathit{sub}(S_2,\tau_{\L,2},\tau_{\R,2}),\dots&\\
			&\hspace{-5em}\nd\symb{\mathit{sub}(S_s,\tau_{\L,s},\tau_{\R,s}),\mathit{here}(\ell,\tau)}\dots}}}\,.
	\end{align*}
	Additionally, we consider the list $\tau_1,\dots,\tau_r$ of all $\psi$-types from $\Pht{\psi}$ (listed in some fixed order), and we define
	\begin{align*}
		\delta(q_0,\ell)=\ell\symb{(q_0,\L),\#\symb{(q_0,\R),\#\symb{\delta(\tau_1,\ell),\#\symb{\delta(\tau_2,\ell),\dots\#\symb{\delta(\tau_r,\ell),\bot}\dots}}}}\,.
	\end{align*}
	For the empty tree we define
	\begin{align*}
		\delta(q_0,\bot)=\bot, \qquad
		\delta(\pht{\psi}{\bot}{\varnothing},\bot)=\bot, \qquad\mbox{and}\qquad
		\delta(\tau,\bot)=\nd_\bot\qquad\mbox{for }\tau\neq\pht{\psi}{\bot}{\varnothing}\,.
	\end{align*}

	The ``main part'' of the result $\Ff(T^1)$, produced using the state $q_0$ is an almost unchanged copy of $T^1$;
	there is only a technical change, that a new $\#$-labeled node is inserted between every node and its right child, so that
	the right child is moved to the left child of this new right child.
	But additionally, below the new $\#$-labeled right child of every node $u$ of $T^1$,
	there are $|\Pht{\psi}|$ modified copies of $T^1\restr u$, attached below a branch of $\#$-labeled nodes (cf.~\cref{fig:trans-1}).
	\begin{figure}
		\centering
		\def\svgscale{0.5}\import{pics/}{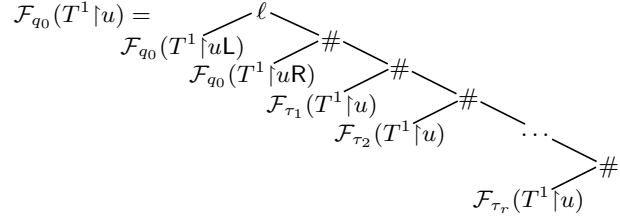}\vspace{-1ex}
		\caption{An illustration of $\Ff_{q_0}(T^1\restr u)$.
			Here, $\ell$ is the label of $u$ in $T^1$, and $\tau_1,\dots,\tau_r$ are all possible $\psi$-types.}
		\label{fig:trans-1}
	\end{figure}
	For each $\psi$-type $\tau$ we have such a copy, namely $\Ff_\tau(T^1\restr u)$, responsible for checking whether the type of $\pi_\AlA(T)\restr u$ can be $\tau$.
	The tree $\Ff_\tau(T^1\restr u)$ is a disjunction (formed by $\nd$-labeled nodes) of all possible triples $(S,\tau_\L,\tau_\R)$ such that
	types $\tau_\L$ and $\tau_\R$ in children of $u$, together with $S$ being the set of those variables among $\X_1,\dots,\X_k$ that contain $u$,
	result in type $\tau$ in $u$ (cf.~\cref{fig:trans-2}).
	\begin{figure}
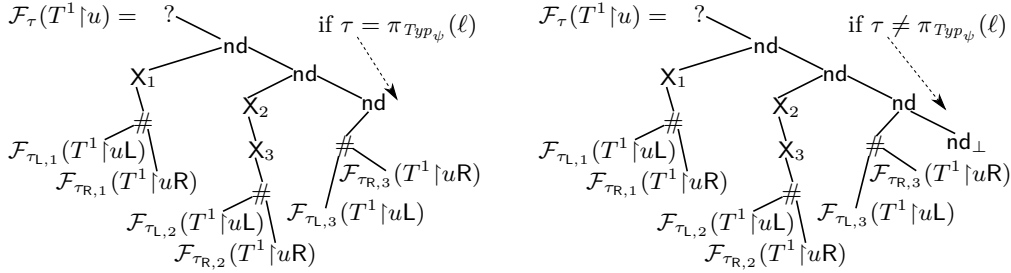

	\centering
	\begin{minipage}[c]{0.49\textwidth}
		\centering
		\def\svgscale{0.5}\import{pics/}{trans-2.pdf_tex_ok}\vspace{-1ex}
	\end{minipage}
	\begin{minipage}[c]{0.49\textwidth}
		\centering
		\def\svgscale{0.5}\import{pics/}{trans-3.pdf_tex_ok}\vspace{-1ex}
	\end{minipage}
	\caption{An illustration of $\Ff_\tau(T^1\restr u)$.
		We assume that there are exactly three triples $(S,\tau_\L,\tau_\R)$ such that $\Comp_{\pi_\AlA(\ell),\psi}(S,\tau_\L,\tau_\R)=\tau$,
		namely $(\set{\X_1},\tau_{\L,1},\tau_{\R,1})$, $(\set{\X_2,\X_3},\tau_{\L,2},\tau_{\R,2})$, and $(\emptyset,\tau_{\L,3},\tau_{\R,3})$,
		for $\ell$ being the label of $u$ in $T^1$.
		We have two cases depending on whether the $\psi$-type written in $\ell$ is $\tau$ or not.}
	\label{fig:trans-2}
	\end{figure}
	We output the variables from $S$ in the resulting tree, so that they can be counted, and then we have subtrees $\Ff_{\tau_\L}(T^1\restr{u\L})$ and $\Ff_{\tau_\R}(T^1\restr{u\R})$,
	responsible for checking whether the type in the children of $u$ can be $\tau_\L$ and $\tau_\R$.
	Additionally, the $\mathit{here}$ subtree allows to finish immediately if $\tau$ is the $\psi$-type of $T^1\restr u$ under the empty valuation.
	Formally, we have the following claim: 

	\begin{restatable}{claim}{claimm}\label{claim}
		For every $\psi$-type $\tau$, numbers $n_1,\dots,n_k\in\Nat$, and node $u$, the following two statements are equivalent:
		\begin{itemize}
		\item	there exist sets $X_1,\dots,X_k$ of nodes of $T\restr u$ such that $\pht{\psi}{\pi_\AlA(T)\restr u}{\varnothing[\X_1\mapsto X_1,\dots,\X_k\mapsto X_k]}=\tau$ and
			$|X_i|=n_i$ for $i\in\interval{1,k}$, and
		\item	there exists a tree $V\in\Ll(\Ff_\tau(T^1\restr u))$ such that $\#_{\X_i}(V)=n_i$ for $i\in\interval{1,k}$.
		\end{itemize}
	\end{restatable}

	\begin{claimproof}
		Let us concentrate on the left-to-right implication.
		The proof is by induction on the maximal depth of nodes in the $X_i$ sets.
		We have three cases.
		First, it is possible that $u$ is not a node of $T$.
		Then, all the sets $X_i$ have to be empty, so we have $\tau=\pht{\psi}{\bot}{\varnothing}$,
		and hence $\Ff_\tau(T^1\restr u)=\delta(\tau,\bot)=\bot$ (recall that $T$ and $T^1$ have the same domain).
		The set $\Ll(\bot)$ contains the tree $\bot$ which indeed has no $\X_i$ labeled nodes, as needed.
	
		Second, it is possible that $u$ is a node of $T$, but all the sets $X_i$ are empty.
		Let $\ell$ be the label of $u$ in $T^1$.
		By construction of $T^1$, we have $\pi_{\Pht{\psi}}(\ell)=\pht{\psi}{\pi_\AlA(T)\restr u}{\varnothing}=\tau$.
		On the rightmost branch of $\Ff_\tau(T^1\restr u)$, after a $?$-labeled node and a few $\nd$-labeled nodes, we have the subtree $\mathit{here}(\ell,\tau)$,
		which is $\bot$ by the above equality.
		We can return the tree $?\symb{\bot,\bot}$, which belongs to $\Ll(\Ff_\tau(T^1\restr u))$.
	
		Finally, suppose that our sets are not all empty.
		Then necessarily $u$ is inside $T$ (and $T^1$); let $\ell$ be the label of $u$ in $T^1$
		(by construction of $T^1$, the label of $u$ in $T$ consists of the first two coordinates of $\ell$).
		Consider $S=\set{\X_i\mid \epsilon\in X_i}$ and $\tau_d=\pht{\psi}{\pi_\AlA(T)\restr{ud}}{\varnothing[\X_1\mapsto X_1,\dots,\X_k\mapsto X_k]\restr d}$ for $d\in\set{\L,\R}$.
		By the induction hypothesis, there are trees $V_d\in\Ll(\Ff_{\tau_d}(T^1\restr{ud}))$ such that $\#_{\X_i}(V_d)=|X_i\restr d|$ for $i\in\interval{1,k}$.
		Due to \cref{eq:compos} we have $\tau=\Comp_{\pi_{\AlA}(\ell),\psi}(S,\tau_\L,\tau_\R)$.
		This means that $\delta(\tau,\ell)$, below a $?$-labeled node and a few $\nd$-labeled, produces a subtree using $\mathit{sub}(S,\tau_\L,\tau_\R)$.
		We define $V$ by choosing this subtree.
		Then, there are some $\X_i$-labeled nodes, for all $\X_i\in S$ (that is, for those sets $X_i$ that contain the root of $T\restr u$).
		Below them, we have the tree $\#\symb{\Ff_{\tau_\L}(T^1\restr{u\L}),\Ff_{\tau_\R}(T^1\restr{u\R})}$;
		in its left subtree we choose $V_\L$, and in its right subtree we choose $V_\R$.
		This way, we obtain a tree $V\in\Ll(\Ff_\tau(T^1\restr u))$, where the number of $\X_i$-labeled nodes is indeed $|X_i|=n_i$, for all $i\in\interval{1,k}$.
	
		We skip the proof of the right-to-left implication, as it is analogous
		(this time, the induction is on the height of the tree $V$).
	\end{claimproof}

	Let $T^2=\Ff(T^1)$.
	As the next operation after $\Ff$, we use $\SUP$.
	Let $T^3=\refl_\SUP(T^2)$.
	The $\SUP$ operation attaches a label to every node of $T^3$ (except for $\nd$-labeled nodes), but we are interested in these labels only in nodes originally (i.e., in $T_2$) labeled by ``$?$''.
	Every such node is the root of a subtree $\refl_\SUP(\Ff_\tau(T^1\restr u))$ for some node $u$ of $T^1$;
	it becomes labeled by $(?,\Uu)$, where $\Uu=\set{A\subseteq\AlA'\mid\SUP_A(\Ll(\Ff_\tau(T^1\restr u)))}$.
	Recall that $\phi=\U(\X_1,\dots,\X_k).\psi$ and that the $\phi$-type is a tuple of $2^k$ coordinates, indexed by sets $I\subseteq\interval{1,k}$.
	Consider such a set $I$, and take $A_I=\set{\X_i\mid i\in I}$.
	By definition of $\SUP_{A_I}$, the label $\Uu$ contains $A_I$ if, and only if, for every $n\in\Nat$
	the language $\Ll(\Ff_\tau(T^1\restr u))$ contains trees with at least $n$ occurrences of every element of $A_I$.
	By \cref{claim} this is the case if, and only if, for every $n\in\Nat$
	there exist sets $X_1,\dots,X_k$ of nodes of $T\restr u$ such that
	$\pht{\psi}{\pi_\AlA(T)\restr u}{\varnothing[\X_1\mapsto X_1,\dots,\X_k\mapsto X_k]}=\tau$ and
	$|X_i|\geq n$ for all $i\in I$.
	This, in turn, holds if, and only if, the $I$-coordinate of the $\phi$-type $\pht{\phi}{\pi_\AlA(T)\restr u}{\varnothing}$ contains $\tau$.
	(The case of $I=\emptyset$ is a bit delicate, but one can see that the proof works without any change also in this case.)

	The above means that all the $\phi$-types we wished to compute are already present in $T^3$, we only have to move them to correct places.
	To this end, for every $\psi$-type $\tau_i$, and for every set $I\subseteq\interval{1,k}$ we append to our sequence of operations a formula $\theta_{i,I}$
	saying that the node $\R^{i+1}\L$ has label of the form $(?,\Uu,\dots)$ with $A_I\in\Uu$
	(note that this node in $\Ff_{q_0}(T^1\restr u)$ is the $?$-labeled root of $\Ff_{\tau_i}(T^1\restr u)$;
	the operation $\SUP$ appends a set $\Uu$ to this label, and operations $\theta_{i',I'}$ applied so far append some additional coordinates that we ignore).

	After that, we already have all $\phi$-types in correct nodes, but in a wrong format;
	we also have additional nodes not present in the original tree $T$.
	To deal with this, at the end we apply a transduction $\Ff'$, which
	\begin{itemize}
	\item	removes all nodes labeled by $(\#,\dots)$ and their right subtrees, hence leaving only nodes present in the original tree $T$;
	\item	the remaining nodes have labels of the form $(a,b,\tau,\Uu,v_{i_1,I_1},\dots,v_{i_s,I_s})$;
		we relabel them to $(a,b,(\set{\tau_i\mid v_{i,I}=\true})_{I\subseteq\interval{1,k}})$.
	\end{itemize}
	This last transduction produces a tree exactly as needed.
\end{proof}

\section{Expressivity}

In this section we prove our second main result, \cref{thm:expressivity}, saying that the simultaneous unboundedness property can be expressed in \wmsouk, but not in \wmsou.
The positive part of this statement is easy:

\begin{proposition}
	For every set of letters $A$ there exists a \wmsouk sentence $\phi$ which holds in a tree $T$ if, and only if, $\SUP_A(\Ll(T))$ holds.
\end{proposition}

\begin{proof}
	Let $A=\set{a_1,\dots,a_k}$.
	We take
	\begin{align*}
		\phi=\U(\X_1,\dots,\X_k).\efin\Y.(a_1(\X_1)\land\dots\land a_k(\X_k)\land\X_1\subseteq\Y\land\dots\land\X_k\subseteq\Y\land\psi(\Y)),
	\end{align*}
	where $\psi(\Y)$ expresses the fact that $\Y$ contains nodes of a single tree from $\Ll(T)$, together with their $\nd$-labeled ancestors, that is,
	that for every node $v$ of $\Y$,
	\begin{itemize}
	\item	the parent of $v$, if exists, belongs to $\Y$;
	\item	if $v$ has label $\nd$, then exactly one child of $v$ belongs to $\Y$
		(strictly speaking: there is a direction $d\in\set{\L,\R}$ such that a child in this direction, if exists, belongs to $\Y$,
		and the child in the opposite direction, if exists, does not belong to $\Y$);
	\item	$v$ does not have label $\nd_\bot$; and
	\item	if $v$ has label other than $\nd$, then all children of $v$ belong to $\Y$.
	\end{itemize}
	It is easy to write the above properties in \wmsouk.
	Then $\phi$ expresses that for every $n\in\Nat$ there exist sets $X_1,\dots,X_k$ of nodes of some $V\in\Ll(T)$ such that $|X_i|\geq n$
	and nodes of $X_i$ have label $a_i$, for all $i\in\interval{1,k}$;
	this is precisely the simultaneous unboundedness property with respect to the set $A=\set{a_1,\dots,a_k}$.
\end{proof}

In the remaining part of this section we prove that SUP with respect to $\set{\a,\b}$ cannot be expressed in \wmsou
(i.e., without using the \U quantifier for tuples of variables).
We prove this already for the word variant of SUP (cf.~\cref{remark1}), which is potentially easier to be expressed than SUP in its full generality.

Our proof is by contradiction.
Assume thus that there is a sentence $\phi_\SUP$ of \wmsou that holds exactly in those trees $T$ for which $\SUP_{\set{\a,\b}}(T)$ is true.
Having $\phi_\SUP$ fixed, we take a number $N$ such that $|\Pht\phi|\leq N$ and $|\Pht{\efin\X.\phi}|\leq N$ for all subformulae $\phi$ of $\phi_\SUP$
(recall that $\Pht\phi$ is a set containing all possible $\phi$-types).

Based on $N$, we now define two trees, $T_1$ and $T_2$, such that $\SUP_{\set{\a,\b}}(T_2)$ but not $\SUP_{\set{\a,\b}}(T_1)$,
and we show that they are indistinguishable by $\phi_\SUP$.
We achieve that by demonstrating their type equality as stated in \cref{T1T2typeequal}, which by \cref{prop:pht-2-form} gives their indistinguishability by the \wmsou sentence $\phi_\SUP$.

\begin{definition}[$T_1$ and $T_2$]\label{def:t1t2}
	\begin{figure}
		\centering
		\def\svgscale{0.5}\import{pics/}{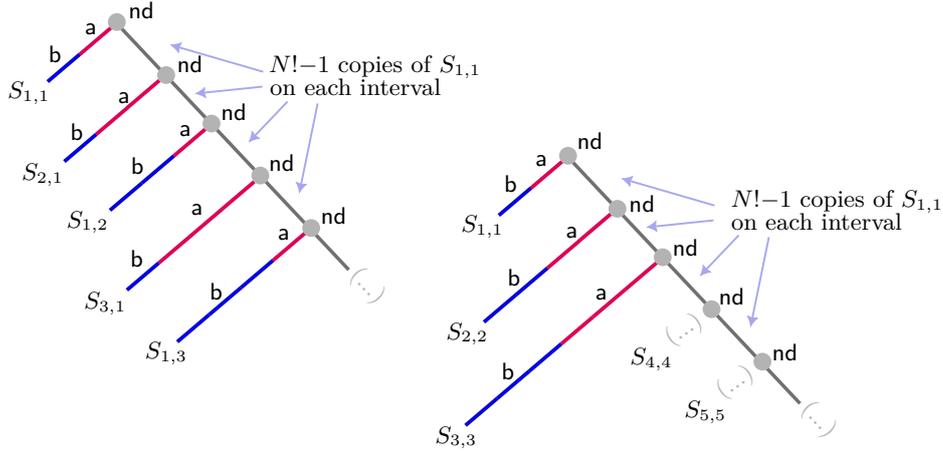}\vspace{-1ex}
		\caption{$T_1$ (left) and $T_2$ (right)}
		\label{fig:T1T2def}
	\end{figure}
	We define $T_1$ as a tree with an infinite rightmost path (that we call its \emph{trunk}), containing $\nd$-labeled nodes.
	For each integer $k \geq 0$, there is a leftward path called \emph{vault} attached to the $(kN!+1)$-th node of the trunk.
	If $k$ is even, we denote the vault as $S_{1, \frac k 2 + 1}$, and otherwise as $S_{\frac{k+1} 2 + 1, 1}$.
	Each vault $S_{m,n}$ consists of two parts: the upper sub-path of length $mN!$, where every node has label $\a$,
	and the lower sub-path of length $nN!$, where every node has label $\b$ (cf.~\cref{fig:T1T2def}).

	To each node of the trunk that does not have a vault attached we attach a copy of $S_{1,1}$.
 	Note that we do not call these copies vaults; only the original $S_{1,1}$ starting at the root of $T_1$ is a vault.

	The definition of $T_2$ is similar to that of $T_1$, except that this time the vault associated with each $k$ is $S_{k,k}$,
	still starting at depth $kN!+1$ and having $kN!$ nodes with label $\a$ followed by $kN!$ nodes with label $\b$.
\end{definition}

The technical core of our proof lies in the following two \lcnamecrefs{lem:T1-T2}:

\begin{restatable}{lemma}{lemTT}\label{lem:T1-T2}
	Let $\psi$ be such that $|\Pht{\efin\X.\psi}|\leq N$.
	If for all $k',\ell'\in\Nat$ we have $\pht{\psi}{T_1\restr\R^{k'N!}}{\varnothing}=\pht{\psi}{T_2\restr\R^{\ell' N!}}{\varnothing}$,
	then for all $k,\ell\in\Nat$ and $\tau\in\Pht\psi$ there exists a function $f\colon\Nat\to\Nat$ such that $\lim_{n\to\infty}f(n)=\infty$ and for all $n\in\Nat$,
	\begin{align*}
		\exists X_1\subseteq\dom(T_1\restr\R^{kN!}).\,|X_1|=n\land\pht{\psi}{T_1\restr\R^{kN!}}{\varnothing[\X\mapsto X_1]}=\tau&\\
		&\hspace{-20em}\Longrightarrow\qquad \exists X_2\subseteq\dom(T_2\restr\R^{\ell N!}).\,f(n)\leq|X_2|<\infty\land\pht{\psi}{T_2\restr\R^{\ell N!}}{\varnothing[\X\mapsto X_2]}=\tau.
	\end{align*}
\end{restatable}

\begin{restatable}{lemma}{lemTTT}\label{lem:T2-T1}
	Let $\psi$ be such that $|\Pht{\efin\X.\psi}|\leq N$.
	If for all $k',\ell'\in\Nat$ we have $\pht{\psi}{T_1\restr\R^{k'N!}}{\varnothing}=\pht{\psi}{T_2\restr\R^{\ell'N!}}{\varnothing}$,
	then for all $k,\ell\in\Nat$ and $\tau\in\Pht\psi$ there exists a function $f\colon\Nat\to\Nat$ such that $\lim_{n\to\infty}f(n)=\infty$ and for all $n\in\Nat$,
	\begin{align*}
		\exists X_1\subseteq\dom(T_1\restr\R^{kN!}).\,f(n)\leq|X_1|<\infty\land\pht{\psi}{T_1\restr\R^{kN!}}{\varnothing[\X\mapsto X_1]} = \tau&\\
		&\hspace{-20em}\Longleftarrow\qquad \exists X_2\subseteq\dom(T_2\restr\R^{\ell N!}).\,|X_2|=n\land\pht{\psi}{T_2\restr\R^{\ell N!}}{\varnothing[\X\mapsto X_2]} = \tau.
	\end{align*}
\end{restatable}

Note that the function $f$ in \cref{lem:T1-T2,lem:T2-T1} may depend on $k$ and $\ell$.
We only sketch here the proof of the above \lcnamecrefs{lem:T2-T1}; a full proof can be found in \cref{app:T1-T2}.

\cref{lem:T1-T2} is slightly easier.
Indeed, suppose first that $k=\ell=0$.
By assumption, in $T_1$ we have a finite set of nodes $X_1$ resulting in a $\psi$-type $\tau$;
based on $X_1$, we have to produce a finite set of nodes $X_2$ in $T_2$, producing the same $\psi$-type $\tau$.
The non-vault nodes of $X_1$ are transferred to $X_2$ without any change; note that the trees $T_1,T_2$ are identical outside of vaults.
When in $T_1$ we have some vault $S_{1,i}$ (or $S_{i,1}$, handled in the same way),
then in the analogous place of $T_2$ we have a vault $S_{j,j}$ with $j\geq i$.
We use a form of pumping to convert $S_{1,i}$ with some nodes marked as elements of $X_1$ into $S_{j,j}$ with marked nodes, which we take to $X_2$;
this is done so that the $\psi$-type does not change.
Namely, we concentrate on $\psi$-types of subtrees of $S_{1,i}$ starting on different levels.
Already in the bottom, $\b$-labeled part of $S_{1,i}$ we can find two levels in distance at most $N$, where the $\psi$-type repeats
(by the pigeonhole principle; recall that the number of possible $\psi$-types is at most $N$).
We then repeat the fragment of $S_{1,i}$ between these two places (together with the set elements marked in it), so that $(j-i)N!$ new nodes are created, and we obtain $S_{1,j}$.
Note that the repeated length, being at most $N$, necessarily divides $N!$.
Because of \cref{lem:compositionality}, such a pumping does not change the $\psi$-type.
In a similar way, we can pump the upper, $\a$-labeled part of $S_{1,j}$, and obtain $S_{j,j}$.
In this way, we convert a finite top part of $T_1$ (with a set $X_1$) into $T_2$ (with a set $X_2$) without changing the $\psi$-type;
the infinite parts located below (where the sets $X_1,X_2$ do not contain any elements) have the same $\psi$-type by the assumption
$\pht{\psi}{T_1\restr\R^{k'N!}}{\varnothing}=\pht{\psi}{T_2\restr\R^{\ell'N!}}{\varnothing}$.
All nodes originally in $X_1$ remained in $X_2$ (possibly shifted), so we have $|X_2|\geq|X_1|$; the \lcnamecref{lem:T1-T2} holds with $f(n)=n$ in this case.

When $k,\ell$ are arbitrary (and we want to change $T_1\restr\R^{kN!}$ into $T_2\restr\R^{\ell N!}$), we proceed in a similar way,
but there is a potential problem that a vault $S_{1,i}$ (or $S_{i,1}$) should be mapped to $S_{j,j}$ with $j<i$;
then we should not stretch the vault, but rather contract it.
But contracting is also possible:
this time we look on $\efin\X.\psi$-types (instead of $\psi$-types) on the shorter target vault $S_{1,j}$;
we can pump the vault as previously, so $S_{1,i}$ and $S_{1,j}$ have the same $\efin\X.\psi$-type.
Because $\efin\X.\psi$-type is a set of all possible $\psi$-types,
we can choose elements of $S_{1,j}$ (and later of $S_{j,j}$) to $X_2$, so that the $\psi$-type is the same as originally in $S_{1,i}$,
although without any guarantees on the size of the new set.
Anyway, the length of vaults in $T_2$ grows two times faster than in $T_1$, so the above problem concerns only the first $\max(0,k-2\ell)$ vaults,
where the number of elements of $X_1$ is bounded by a constant $c_{k,\ell}$ (depending on $k$ and $\ell$).
All further elements of $X_1$ contribute to the size of $X_2$; the \lcnamecref{lem:T1-T2} holds with $f(n)=n-c_{k,\ell}$.

Consider now \cref{lem:T2-T1}, where we have to create a set $X_1$ in $T_1$ based on a set $X_2$ in $T_2$.
There are two cases.
Suppose first that at least half of elements of $X_2$ lie outside of the vaults.
In this case we proceed as previously, appropriately stretching and/or contracting the vaults.
While there is no size guarantee for vault elements, already by counting elements outside of the vaults we obtain $|X_2|\geq\frac{|X_1|}{2}$.

In the opposite case, we check which label is more frequent among the (at least $\frac{|X_2|}{2}$) vault elements of $X_2$.
Suppose this is $\a$ (the case of $\b$ is analogous), and that $k=\ell=0$.
We then map every vault $S_{i,i}$ into $S_{i,1}$, contracting only the $\b$-labeled part; all the $\a$-labeled vault elements of $X_2$ remain in $X_1$.
Because the distance between $S_{i,i}$ and $S_{i+1,i+1}$ in $T_2$ is $N!$, while the distance between $S_{i,1}$ and $S_{i+1,1}$ in $T_1$ is $2N!$,
we also need to stretch the trunk, which is possible using a similar pumping argument
(and we stretch some $S_{1,1}$ into the vault $S_{1,i}$ that should be in the middle between $S_{i,1}$ and $S_{i+1,1}$).

This is almost the end, except that we need to handle arbitrary $k,\ell$.
To this end, we either stretch the initial fragment of the trunk of length $N!$ into multiple such fragments,
or we contract the initial fragment of appropriate length into a fragment of length $N!$, so that the vault lengths become synchronized.

Having \cref{lem:T1-T2,lem:T2-T1} we can conclude that the trees $T_1$ and $T_2$ (cf.\@ \cref{def:t1t2}) have the same types:

\begin{lemma}\label{T1T2typeequal}
	Let $\phi$ be a subformula of $\phi_\SUP$.
	Then for all $k,\ell\in\Nat$ we have $\pht{\phi}{T_1\restr\R^{kN!}}{\varnothing} = \pht{\phi}{T_2\restr\R^{\ell N!}}{\varnothing}$.
\end{lemma}

\begin{proof}
	We proceed by induction on $\phi$, considering all possible forms of the formula.
	First, note that we only consider the valuation $\varnothing$, mapping all variables to the empty set, so for atomic formulae of the form $a(\X)$ or $\X\subseteq\Y$
	the $\phi$-type is always $\true$, and for $\X\child_d\Y$ the $\phi$-type is always $\mathsf{empty}$.
	For $\phi=\psi_1\land\psi_2$ the $\phi$-type is just the pair containing the $\psi_1$-type and the $\psi_2$-type;
	for them we have the equality $\pht{\psi_i}{T_1\restr\R^{kN!}}{\varnothing} = \pht{\psi_i}{T_2\restr\R^{\ell N!}}{\varnothing}$ by the induction hypothesis.
	Likewise, for $\phi=\neg\psi$ the $\phi$-type equals the $\psi$-type, and we immediately conclude by the induction hypothesis
	$\pht{\psi}{T_1\restr\R^{kN!}}{\varnothing} = \pht{\psi}{T_2\restr\R^{\ell N!}}{\varnothing}$.
	
	Suppose now that $\phi=\efin\X.\psi$.
	Then the $\phi$-type of $T_1\restr\R^{kN!}$ is the set of $\psi$-types $\pht{\psi}{T_1\restr\R^{kN!}}{\varnothing[\X\mapsto X_1]}$
	over all possible finite sets $X_1\subseteq\dom(T_1\restr\R^{kN!})$, and likewise for $T_2$.
	By \cref{lem:T1-T2}, for every $\psi$-type of $T_1\restr\R^{kN!}$ there exists a set $X_2$
	giving the same $\psi$-type for $T_2\restr\R^{\ell N!}$, and conversely by \cref{lem:T2-T1},
	so the two $\phi$-types are equal (recall that $N$ was chosen such that $|\Pht{\efin\X.\psi}|\leq N$ whenever $\psi$ is a subformula of $\phi_\SUP$,
	hence the two \lcnamecrefs{lem:T1-T2} can indeed be applied).
	
	Finally, suppose that $\phi=\U\X.\psi$.
	Then the $\phi$-type consists of two coordinates.
	On the first coordinate we simply have the $\efin\X.\psi$-type---these types are equal for $T_1\restr\R^{kN!}$ and $T_2\restr\R^{\ell N!}$ by the previous case.
	On the second coordinate we have the set of $\psi$-types $\tau$ such that
	$\pht{\psi}{T_1\restr\R^{kN!}}{\varnothing[\X\mapsto X_1]}=\tau$ for arbitrarily large finite sets $X_1$, and likewise for $X_2$.
	But $\pht{\psi}{T_1\restr\R^{kN!}}{\varnothing[\X\mapsto X_1]}=\tau$ for arbitrarily large finite sets $X_1$ if and only if
	$\pht{\psi}{T_2\restr\R^{\ell N!}}{\varnothing[\X\mapsto X_2]}=\tau$ for arbitrarily large finite sets $X_2$, by \cref{lem:T1-T2,lem:T2-T1}.
	This gives us equality of the two $\phi$-types.
	
	Recall that by assumption $\phi_\SUP$ is a formula of \wmsou, without quantification over tuples, so the above exhausts all possible cases.
\end{proof}

\cref{T1T2typeequal} implies in particular that $\pht{\phi_\SUP}{T_1}{\varnothing}=\pht{\phi_\SUP}{T_2}{\varnothing}$,
which by \cref{prop:pht-2-form} means that $\phi_\SUP$ is satisfied in $T_1$ if and only if it is satisfied in $T_2$.
This way we reach a contradiction with the fact that $\phi_\SUP$ should be true in $T_1$, but not in $T_2$.
Thus, the simultaneous unboundedness property for two letters cannot be expressed by a formula $\phi_\SUP$ not involving the \U quantifiers for tuples of variables;
we obtain \cref{thm:expressivity}.

\begin{remark}\label{remark3}
	We have shown that SUP with respect to a two-element set $\set{\a,\b}$ cannot be expressed without quantification over pairs of variables.
	It is easy to believe that using a very similar proof one can show that
	SUP with respect to a $k$-element set cannot be expressed without quantification over $k$-tuples of variables, for every $k\geq 2$.
\end{remark}

\bibliographystyle{plainurl}
\bibliography{bib}

\begin{thebibliography}{10}

\bibitem{AhoIndexed}
Alfred~V. Aho.
\newblock Indexed grammars - an extension of context-free grammars.
\newblock {\em J. {ACM}}, 15(4):647--671, 1968.
\newblock \href {https://doi.org/10.1145/321479.321488}
  {\path{doi:10.1145/321479.321488}}.

\bibitem{wcmso-safe-schemes}
David Barozzini, Lorenzo Clemente, Thomas Colcombet, and Paweł Parys.
\newblock Cost automata, safe schemes, and downward closures.
\newblock In Artur Czumaj, Anuj Dawar, and Emanuela Merelli, editors, {\em 47th
  International Colloquium on Automata, Languages, and Programming, {ICALP}
  2020, July 8-11, 2020, Saarbr{\"{u}}cken, Germany (Virtual Conference)},
  volume 168 of {\em LIPIcs}, pages 109:1--109:18. Schloss Dagstuhl -
  Leibniz-Zentrum f{\"{u}}r Informatik, 2020.
\newblock \href {https://doi.org/10.4230/LIPIcs.ICALP.2020.109}
  {\path{doi:10.4230/LIPIcs.ICALP.2020.109}}.

\bibitem{BojanczykU}
Mikołaj Bojańczyk.
\newblock A bounding quantifier.
\newblock In Jerzy Marcinkowski and Andrzej Tarlecki, editors, {\em Computer
  Science Logic, 18th International Workshop, {CSL} 2004, 13th Annual
  Conference of the EACSL, Karpacz, Poland, September 20-24, 2004,
  Proceedings}, volume 3210 of {\em Lecture Notes in Computer Science}, pages
  41--55. Springer, 2004.
\newblock \href {https://doi.org/10.1007/978-3-540-30124-0_7}
  {\path{doi:10.1007/978-3-540-30124-0_7}}.

\bibitem{wmsou-words}
Mikołaj Bojańczyk.
\newblock Weak {MSO} with the unbounding quantifier.
\newblock {\em Theory Comput. Syst.}, 48(3):554--576, 2011.
\newblock \href {https://doi.org/10.1007/s00224-010-9279-2}
  {\path{doi:10.1007/s00224-010-9279-2}}.

\bibitem{wmsou-path}
Mikołaj Bojańczyk.
\newblock Weak {MSO+U} with path quantifiers over infinite trees.
\newblock In Javier Esparza, Pierre Fraigniaud, Thore Husfeldt, and Elias
  Koutsoupias, editors, {\em Automata, Languages, and Programming - 41st
  International Colloquium, {ICALP} 2014, Copenhagen, Denmark, July 8-11, 2014,
  Proceedings, Part {II}}, volume 8573 of {\em Lecture Notes in Computer
  Science}, pages 38--49. Springer, 2014.
\newblock \href {https://doi.org/10.1007/978-3-662-43951-7_4}
  {\path{doi:10.1007/978-3-662-43951-7_4}}.

\bibitem{bounds}
Mikołaj Bojańczyk and Thomas Colcombet.
\newblock Bounds in $\omega$-regularity.
\newblock In {\em 21th {IEEE} Symposium on Logic in Computer Science {(LICS}
  2006), 12-15 August 2006, Seattle, WA, USA, Proceedings}, pages 285--296.
  {IEEE} Computer Society, 2006.
\newblock \href {https://doi.org/10.1109/LICS.2006.17}
  {\path{doi:10.1109/LICS.2006.17}}.

\bibitem{wmsou-trees}
Mikołaj Bojańczyk and Szymon Toruńczyk.
\newblock Weak {MSO+U} over infinite trees.
\newblock In Christoph D{\"{u}}rr and Thomas Wilke, editors, {\em 29th
  International Symposium on Theoretical Aspects of Computer Science, {STACS}
  2012, February 29th - March 3rd, 2012, Paris, France}, volume~14 of {\em
  LIPIcs}, pages 648--660. Schloss Dagstuhl - Leibniz-Zentrum f{\"{u}}r
  Informatik, 2012.
\newblock \href {https://doi.org/10.4230/LIPIcs.STACS.2012.648}
  {\path{doi:10.4230/LIPIcs.STACS.2012.648}}.

\bibitem{BrazdilCKN12}
Tom{\'{a}}s Br{\'{a}}zdil, Krishnendu Chatterjee, Anton{\'{\i}}n Kucera, and
  Petr Novotn{\'{y}}.
\newblock Efficient controller synthesis for consumption games with multiple
  resource types.
\newblock In P.~Madhusudan and Sanjit~A. Seshia, editors, {\em Computer Aided
  Verification - 24th International Conference, {CAV} 2012, Berkeley, CA, USA,
  July 7-13, 2012 Proceedings}, volume 7358 of {\em Lecture Notes in Computer
  Science}, pages 23--38. Springer, 2012.
\newblock \href {https://doi.org/10.1007/978-3-642-31424-7_8}
  {\path{doi:10.1007/978-3-642-31424-7_8}}.

\bibitem{OrderedMultiPushdown}
Luca Breveglieri, Alessandra Cherubini, Claudio Citrini, and Stefano
  Crespi{-}Reghizzi.
\newblock Multi-push-down languages and grammars.
\newblock {\em Int. J. Found. Comput. Sci.}, 7(3):253--292, 1996.
\newblock \href {https://doi.org/10.1142/S0129054196000191}
  {\path{doi:10.1142/S0129054196000191}}.

\bibitem{reflection-selection}
Christopher~H. Broadbent, Arnaud Carayol, C.{-}H.~Luke Ong, and Olivier Serre.
\newblock Higher-order recursion schemes and collapsible pushdown automata:
  Logical properties.
\newblock {\em {ACM} Trans. Comput. Log.}, 22(2):12:1--12:37, 2021.
\newblock \href {https://doi.org/10.1145/3452917} {\path{doi:10.1145/3452917}}.

\bibitem{Buchi}
Julius~Richard B{\"u}chi.
\newblock On a decision method in restricted second order arithmetic.
\newblock In {\em Proceedings of the 1960 International Congress on Logic,
  Methodology and Philosophy of Science}, pages 1--11. Stanford University
  Press, 1962.

\bibitem{CarapelleKL13}
Claudia Carapelle, Alexander Kartzow, and Markus Lohrey.
\newblock Satisfiability of {CTL}* with constraints.
\newblock In Pedro~R. D'Argenio and Hern{\'{a}}n~C. Melgratti, editors, {\em
  {CONCUR} 2013 - Concurrency Theory - 24th International Conference, {CONCUR}
  2013, Buenos Aires, Argentina, August 27-30, 2013. Proceedings}, volume 8052
  of {\em Lecture Notes in Computer Science}, pages 455--469. Springer, 2013.
\newblock \href {https://doi.org/10.1007/978-3-642-40184-8_32}
  {\path{doi:10.1007/978-3-642-40184-8_32}}.

\bibitem{Ordered-Tree-Pushdown}
Lorenzo Clemente, Paweł Parys, Sylvain Salvati, and Igor Walukiewicz.
\newblock Ordered tree-pushdown systems.
\newblock In Prahladh Harsha and G.~Ramalingam, editors, {\em 35th {IARCS}
  Annual Conference on Foundation of Software Technology and Theoretical
  Computer Science, {FSTTCS} 2015, December 16-18, 2015, Bangalore, India},
  volume~45 of {\em LIPIcs}, pages 163--177. Schloss Dagstuhl - Leibniz-Zentrum
  f{\"{u}}r Informatik, 2015.
\newblock \href {https://doi.org/10.4230/LIPIcs.FSTTCS.2015.163}
  {\path{doi:10.4230/LIPIcs.FSTTCS.2015.163}}.

\bibitem{diagonal}
Lorenzo Clemente, Paweł Parys, Sylvain Salvati, and Igor Walukiewicz.
\newblock The diagonal problem for higher-order recursion schemes is decidable.
\newblock In Martin Grohe, Eric Koskinen, and Natarajan Shankar, editors, {\em
  Proceedings of the 31st Annual {ACM/IEEE} Symposium on Logic in Computer
  Science, {LICS} '16, New York, NY, USA, July 5-8, 2016}, pages 96--105.
  {ACM}, 2016.
\newblock \href {https://doi.org/10.1145/2933575.2934527}
  {\path{doi:10.1145/2933575.2934527}}.

\bibitem{regular-cost-functions}
Thomas Colcombet.
\newblock The theory of stabilisation monoids and regular cost functions.
\newblock In Susanne Albers, Alberto Marchetti{-}Spaccamela, Yossi Matias,
  Sotiris~E. Nikoletseas, and Wolfgang Thomas, editors, {\em Automata,
  Languages and Programming, 36th Internatilonal Colloquium, {ICALP} 2009,
  Rhodes, Greece, July 5-12, 2009, Proceedings, Part {II}}, volume 5556 of {\em
  Lecture Notes in Computer Science}, pages 139--150. Springer, 2009.
\newblock \href {https://doi.org/10.1007/978-3-642-02930-1_12}
  {\path{doi:10.1007/978-3-642-02930-1_12}}.

\bibitem{sep-piecewise-test}
Wojciech Czerwiński, Wim Martens, Lorijn van Rooijen, Marc Zeitoun, and Georg
  Zetzsche.
\newblock A characterization for decidable separability by piecewise testable
  languages.
\newblock {\em Discret. Math. Theor. Comput. Sci.}, 19(4), 2017.
\newblock \href {https://doi.org/10.23638/DMTCS-19-4-1}
  {\path{doi:10.23638/DMTCS-19-4-1}}.

\bibitem{Damm82}
Werner Damm.
\newblock The {IO-} and {OI}-hierarchies.
\newblock {\em Theor. Comput. Sci.}, 20:95--207, 1982.
\newblock \href {https://doi.org/10.1016/0304-3975(82)90009-3}
  {\path{doi:10.1016/0304-3975(82)90009-3}}.

\bibitem{Elgot}
Calvin~C. Elgot.
\newblock Decision problems of finite automata design and related arithmetics.
\newblock {\em Trans. Amer. Math. Soc.}, 98(1):21--51, 1961.
\newblock \href {https://doi.org/10.2307/1993511} {\path{doi:10.2307/1993511}}.

\bibitem{EmersonJutla}
E.~Allen Emerson and Charanjit~S. Jutla.
\newblock Tree automata, mu-calculus and determinacy (extended abstract).
\newblock In {\em 32nd Annual Symposium on Foundations of Computer Science, San
  Juan, Puerto Rico, 1-4 October 1991}, pages 368--377. {IEEE} Computer
  Society, 1991.
\newblock \href {https://doi.org/10.1109/SFCS.1991.185392}
  {\path{doi:10.1109/SFCS.1991.185392}}.

\bibitem{FijalkowZ12}
Nathana{\"{e}}l Fijalkow and Martin Zimmermann.
\newblock Cost-parity and cost-{S}treett games.
\newblock In Deepak D'Souza, Telikepalli Kavitha, and Jaikumar Radhakrishnan,
  editors, {\em {IARCS} Annual Conference on Foundations of Software Technology
  and Theoretical Computer Science, {FSTTCS} 2012, December 15-17, 2012,
  Hyderabad, India}, volume~18 of {\em LIPIcs}, pages 124--135. Schloss
  Dagstuhl - Leibniz-Zentrum f{\"{u}}r Informatik, 2012.
\newblock \href {https://doi.org/10.4230/LIPIcs.FSTTCS.2012.124}
  {\path{doi:10.4230/LIPIcs.FSTTCS.2012.124}}.

\bibitem{diagonal-safe}
Matthew Hague, Jonathan Kochems, and C.{-}H.~Luke Ong.
\newblock Unboundedness and downward closures of higher-order pushdown
  automata.
\newblock In Rastislav Bod{\'{\i}}k and Rupak Majumdar, editors, {\em
  Proceedings of the 43rd Annual {ACM} {SIGPLAN-SIGACT} Symposium on Principles
  of Programming Languages, {POPL} 2016, St. Petersburg, FL, USA, January 20 -
  22, 2016}, pages 151--163. {ACM}, 2016.
\newblock \href {https://doi.org/10.1145/2837614.2837627}
  {\path{doi:10.1145/2837614.2837627}}.

\bibitem{collapsible}
Matthew Hague, Andrzej~S. Murawski, C.{-}H.~Luke Ong, and Olivier Serre.
\newblock Collapsible pushdown automata and recursion schemes.
\newblock volume~18, pages 25:1--25:42, 2017.
\newblock \href {https://doi.org/10.1145/3091122} {\path{doi:10.1145/3091122}}.

\bibitem{KNU-hopda}
Teodor Knapik, Damian Niwiński, and Paweł Urzyczyn.
\newblock Higher-order pushdown trees are easy.
\newblock In Mogens Nielsen and Uffe Engberg, editors, {\em Foundations of
  Software Science and Computation Structures, 5th International Conference,
  {FOSSACS} 2002. Held as Part of the Joint European Conferences on Theory and
  Practice of Software, {ETAPS} 2002 Grenoble, France, April 8-12, 2002,
  Proceedings}, volume 2303 of {\em Lecture Notes in Computer Science}, pages
  205--222. Springer, 2002.
\newblock \href {https://doi.org/10.1007/3-540-45931-6_15}
  {\path{doi:10.1007/3-540-45931-6_15}}.

\bibitem{KobayashiACM}
Naoki Kobayashi.
\newblock Model checking higher-order programs.
\newblock {\em J. {ACM}}, 60(3):20:1--20:62, 2013.
\newblock \href {https://doi.org/10.1145/2487241.2487246}
  {\path{doi:10.1145/2487241.2487246}}.

\bibitem{KobayashiOngtypes}
Naoki Kobayashi and C.{-}H.~Luke Ong.
\newblock A type system equivalent to the modal mu-calculus model checking of
  higher-order recursion schemes.
\newblock In {\em Proceedings of the 24th Annual {IEEE} Symposium on Logic in
  Computer Science, {LICS} 2009, 11-14 August 2009, Los Angeles, CA, {USA}},
  pages 179--188. {IEEE} Computer Society, 2009.
\newblock \href {https://doi.org/10.1109/LICS.2009.29}
  {\path{doi:10.1109/LICS.2009.29}}.

\bibitem{KupfermanPitermanVardi09}
Orna Kupferman, Nir Piterman, and Moshe~Y. Vardi.
\newblock From liveness to promptness.
\newblock {\em Formal Methods Syst. Des.}, 34(2):83--103, 2009.
\newblock \href {https://doi.org/10.1007/s10703-009-0067-z}
  {\path{doi:10.1007/s10703-009-0067-z}}.

\bibitem{Ong-hoschemes}
C.{-}H.~Luke Ong.
\newblock On model-checking trees generated by higher-order recursion schemes.
\newblock In {\em 21th {IEEE} Symposium on Logic in Computer Science {(LICS}
  2006), 12-15 August 2006, Seattle, WA, USA, Proceedings}, pages 81--90.
  {IEEE} Computer Society, 2006.
\newblock \href {https://doi.org/10.1109/LICS.2006.38}
  {\path{doi:10.1109/LICS.2006.38}}.

\bibitem{wmsou-schemes}
Paweł Parys.
\newblock Recursion schemes and the {WMSO+U} logic.
\newblock In Rolf Niedermeier and Brigitte Vall{\'{e}}e, editors, {\em 35th
  Symposium on Theoretical Aspects of Computer Science, {STACS} 2018, February
  28 to March 3, 2018, Caen, France}, volume~96 of {\em LIPIcs}, pages
  53:1--53:16. Schloss Dagstuhl - Leibniz-Zentrum f{\"{u}}r Informatik, 2018.
\newblock \href {https://doi.org/10.4230/LIPIcs.STACS.2018.53}
  {\path{doi:10.4230/LIPIcs.STACS.2018.53}}.

\bibitem{wmsou-schemes-journal}
Paweł Parys.
\newblock Recursion schemes, the {MSO} logic, and the {U} quantifier.
\newblock {\em Log. Methods Comput. Sci.}, 16(1), 2020.
\newblock \href {https://doi.org/10.23638/LMCS-16(1:20)2020}
  {\path{doi:10.23638/LMCS-16(1:20)2020}}.

\bibitem{types-diagonal-journal}
Paweł Parys.
\newblock A type system describing unboundedness.
\newblock {\em Discret. Math. Theor. Comput. Sci.}, 22(4), 2020.
\newblock \href {https://doi.org/10.23638/DMTCS-22-4-2}
  {\path{doi:10.23638/DMTCS-22-4-2}}.

\bibitem{RabinMSO}
Michael~O. Rabin.
\newblock Decidability of second-order theories and automata on infinite trees.
\newblock {\em Trans. Amer. Math. Soc.}, 141:1--35, 1969.
\newblock \href {https://doi.org/10.2307/1995086} {\path{doi:10.2307/1995086}}.

\bibitem{KrivineWS}
Sylvain Salvati and Igor Walukiewicz.
\newblock Krivine machines and higher-order schemes.
\newblock {\em Inf. Comput.}, 239:340--355, 2014.
\newblock \href {https://doi.org/10.1016/j.ic.2014.07.012}
  {\path{doi:10.1016/j.ic.2014.07.012}}.

\bibitem{ModelSW}
Sylvain Salvati and Igor Walukiewicz.
\newblock A model for behavioural properties of higher-order programs.
\newblock In Stephan Kreutzer, editor, {\em 24th {EACSL} Annual Conference on
  Computer Science Logic, {CSL} 2015, September 7-10, 2015, Berlin, Germany},
  volume~41 of {\em LIPIcs}, pages 229--243. Schloss Dagstuhl - Leibniz-Zentrum
  f{\"{u}}r Informatik, 2015.
\newblock \href {https://doi.org/10.4230/LIPIcs.CSL.2015.229}
  {\path{doi:10.4230/LIPIcs.CSL.2015.229}}.

\bibitem{schemes-lY}
Sylvain Salvati and Igor Walukiewicz.
\newblock Simply typed fixpoint calculus and collapsible pushdown automata.
\newblock {\em Math. Struct. Comput. Sci.}, 26(7):1304--1350, 2016.
\newblock \href {https://doi.org/10.1017/S0960129514000590}
  {\path{doi:10.1017/S0960129514000590}}.

\bibitem{Trakhtenbrot}
Boris Trakhtenbrot.
\newblock Finite automata and the logic of monadic predicates.
\newblock {\em Doklady Akademii Nauk SSSR}, 140:326--329, 1961.

\bibitem{Zetzsche-dc}
Georg Zetzsche.
\newblock An approach to computing downward closures.
\newblock In Magn{\'{u}}s~M. Halld{\'{o}}rsson, Kazuo Iwama, Naoki Kobayashi,
  and Bettina Speckmann, editors, {\em Automata, Languages, and Programming -
  42nd International Colloquium, {ICALP} 2015, Kyoto, Japan, July 6-10, 2015,
  Proceedings, Part {II}}, volume 9135 of {\em Lecture Notes in Computer
  Science}, pages 440--451. Springer, 2015.
\newblock \href {https://doi.org/10.1007/978-3-662-47666-6_35}
  {\path{doi:10.1007/978-3-662-47666-6_35}}.

\end{thebibliography}
\newpage\appendix





\section{Proof of Lemmata \ref{lem:T1-T2} and \ref{lem:T2-T1}}\label{app:T1-T2}

Before starting the actual proof, we have a sequence of \lcnamecrefs{exlemma1} providing equality of types in some situations.
The first \lcnamecref{exlemma1} is an immediate consequence of \cref{lem:compositionality}:

\begin{lemma}\label{exlemma1}
	Let $U_1,U_2$ be trees, and let $\nu_1,\nu_2$ be valuations.
	If $\epsilon \in \nu_1(\X) \iff \epsilon \in \nu_2(\X)$ for all variables $\X$, and
	$\pht{\phi}{U_1\restr\L}{\nu_1 \restr\L} = \pht{\phi}{U_2\restr\L}{\nu_2\restr\L}$, and $\pht{\phi}{U_1\restr\R}{\nu_1 \restr\R} = \pht{\phi}{U_2\restr\R}{\nu_2 \restr\R}$,
	then $\pht{\phi}{U_1}{\nu_1} = \pht{\phi}{U_2}{\nu_2}$.
\end{lemma}

Next, we have a form of pumping:

\begin{figure}
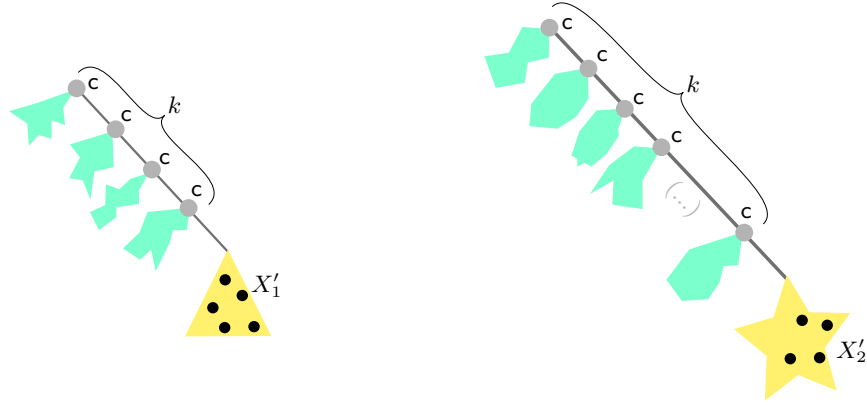

	\centering
	\begin{minipage}[c]{0.49\textwidth}
		\centering
		\def\svgscale{0.5}\import{pics/}{lemma5Left.pdf_tex_ok}\vspace{-1ex}
	\end{minipage}
	\begin{minipage}[c]{0.49\textwidth}
		\centering
		\def\svgscale{0.5}\import{pics/}{lemma5Right.pdf_tex_ok}\vspace{-1ex}
	\end{minipage}
	\caption{Trees $U_1$ (left) and $U_2$ (right) from \cref{MainLemmaForULeftLemma5}.
		All green subtrees have the same $\efin\X.\psi$-type.
		The two yellow subtrees, along with the sets $X_1',X_2'$ in them, have the same $\psi$-type.}
	\label{fig:lemma5}
\end{figure}

\begin{lemma}\label{MainLemmaForULeftLemma5}
	Let $\psi$ be such that $|\Pht{\efin\X.\psi}|\leq N$.
	Let $U_1$ and $U_2$ be trees and let $X_1',X_2'$ be finite sets of nodes such that, for some $k\geq N$ (cf.~\cref{fig:lemma5}),
	\begin{itemize}
	\item	nodes $\R^{i-1}$ of $U_1$ for $i\in[1,k]$ and nodes $\R^{i-1}$ of $U_2$ for $i\in[1,k+N!]$ have all the same label (these two paths are called the \emph{trunks}),
	\item	subtrees $U_1\restr{\R^{i-1}\L}$ for $i\in[1,k]$ and subtrees $U_2\restr{\R^{i-1}\L}$ for $i\in[1,k+N!]$ (i.e., to the left of the trunks) have all the same $\efin\X.\psi$-type, and
	\item	$\pht{\psi}{U_1\restr{\R^k}}{\varnothing[\X \mapsto X'_1]} = \pht{\psi}{U_2\restr{\R^{k+N!}}}{\varnothing[\X \mapsto X'_2]}$.
	\end{itemize}
	Then
	\begin{enumerate}
	\item	for every finite set $X_1\subseteq\dom(U_1)$ such that $X_1\restr{\R^k}=X'_1$,
		there exists a finite set $X_2\subseteq\dom(U_2)$ such that $X_2\restr{\R^{k+N!}}=X'_2$,
		and $\pht{\psi}{U_1}{\varnothing[\X\mapsto X_1]}=\pht{\psi}{U_2}{\varnothing[\X\mapsto X_2]}$, 
		and $|X_1\cap\set{\R^{i-1}\mid i\in[1,k]}|\leq|X_2\cap\set{\R^{i-1}\mid i\in[1,k+N!]}|$,
		and, conversely
	\item	for every finite set $X_2\subseteq\dom(U_2)$ such that $X_2\restr{\R^{k+N!}}=X'_2$,
		there exists a finite set $X_1\subseteq\dom(U_1)$ such that $X_1\restr{\R^k}=X'_1$ 
		and $\pht{\psi}{U_1}{\varnothing[\X\mapsto X_1]}=\pht{\psi}{U_2}{\varnothing[\X\mapsto X_2]}$.
	\end{enumerate}
\end{lemma}

\begin{proof}
	Concentrate first on Item 1.
	As a first step, we construct a tree $U_{1.5}$ together with a set $X_{1.5}$, starting from $U_1$ and $X_1$.
	To this end, for each level $i\in[0,N]$ we consider the $\psi$-type $\pht{\psi}{U_1\restr{\R^i}}{\varnothing[\X\mapsto X_1\restr{\R^i}]}\in\Pht\psi$.
	Because $N\geq|\Pht{\efin\X.\psi}|\geq|\Pht\psi|$, we can find a pair of levels $\ell,m$ (with $0\leq\ell<m\leq N\leq k$)
	such that $\pht{\psi}{U_1\restr{\R^\ell}}{\varnothing[\X\mapsto X_1\restr{\R^\ell}]}=\pht{\psi}{U_1 \restr{\R^m}}{\varnothing[\X\mapsto X_1\restr{\R^m}]}$.
	We then replicate $\frac{N!}{m-\ell}+1$ times the fragment of $U_1$ lying between the nodes $\R^\ell$ (inclusively) and $\R^m$ (exclusively), naming the resulting tree $U_{1.5}$
	(importantly $m-\ell$ is at most $N$, so it divides $N!$).
	Simultaneously we create a set $X_{1.5}$ by taking $X_1$ and adding to it the copies of the nodes from the replicated fragments that belonged to $X_1$
	(of course the part of $X_1$ containing descendants of $\R^m$ is now shifted appropriately, together with the subtree).
	Repeated application of \cref{exlemma1} then implies that $\pht{\psi}{U_{1.5}}{\varnothing[\X \mapsto X_{1.5}]} = \pht{\psi}{U_1}{\varnothing[\X \mapsto X_1]}$
	(for every $i\in[0,\frac{N!}{m-\ell}+1]$, on level $\ell+i(m-\ell)$ we obtain the same $\psi$-type as originally on levels $\ell$ and $m$).

	As a second step, from $X_{1.5}$ we construct $X_2\subseteq\dom(U_2)$.
	The part $X_2\restr{\R^{k+N!}}$ is already specified to be $X'_2$,
	which by the assumption $\pht{\psi}{U_1\restr{\R^k}}{\varnothing[\X \mapsto X'_1]} = \pht{\psi}{U_2\restr{\R^{k+N!}}}{\varnothing[\X \mapsto X'_2]}$
	guarantees $\pht{\psi}{U_{1.5}\restr{\R^{k+N!}}}{\varnothing[\X \mapsto X_{1.5}\restr{\R^{k+N!}}]} = \pht{\psi}{U_2\restr{\R^{k+N!}}}{\varnothing[\X \mapsto X_2\restr{\R^{k+N!}}]}$.
	Among nodes $\R^{i-1}$ for $i\in[1,k+N!]$ (i.e., the trunk) we take to $X_2$ exactly those that were in $X_{1.5}$.
	Consider now a subtree rooted in $\R^{i-1}\L$ for $i\in[1,k+N!]$.
	By assumption, $U_{1.5}\restr{\R^{i-1}\L}$ and $U_2\restr{\R^{i-1}\L}$ have the same $\efin\X.\psi$-type.
	Such a $\efin\X.\psi$-type (for, say, $U_{1.5}$) is defined as a set containing $\psi$-types $\pht{\psi}{U_{1.5}\restr{\R^{i-1}\L}}{\varnothing[\X\mapsto X'']}$
	for all finite sets $X''\subseteq\dom(U_{1.5}\restr{\R^{i-1}\L})$;
	in particular, it contains $\pht{\psi}{U_{1.5}\restr{\R^{i-1}\L}}{\varnothing[\X\mapsto X_{1.5}\restr{\R^{i-1}\L}]}$.
	The equality of $\efin\X.\psi$-types implies that there is also a finite set $X''\subseteq\dom(U_2\restr{\R^{i-1}\L})$ such that
	$\pht{\psi}{U_{1.5}\restr{\R^{i-1}\L}}{\varnothing[\X\mapsto X_{1.5}\restr{\R^{i-1}\L}]}=\pht{\psi}{U_2\restr{\R^{i-1}\L}}{\varnothing[\X\mapsto X'']}$.
	We then choose the part $X_2\restr{\R^{i-1}\L}$ of $X_2$ to be equal $X''$.
	In consequence, the trunks of the two trees look the same, and the subtrees below the trunks have the same $\psi$-types;
	a repeated use of \cref{exlemma1} gives us then $\pht{\psi}{U_{1.5}}{\varnothing[\X \mapsto X_{1.5}]} = \pht{\psi}{U_2}{\varnothing[\X \mapsto X_2]}$.

	Following the established equalities, we conclude that $\pht{\psi}{U_1}{\varnothing[\X \mapsto X_1]} = \pht{\psi}{U_2}{\varnothing[\X \mapsto X_2]}$.
	When counting nodes of the trunks that are in the sets $X_1,X_{1.5},X_2$, we see that their number between $X_1$ and $X_{1.5}$ can only increase
	(we repeat some fragments of $X_1$) and later it remains unchanged; this gives us the inequality required by the statement of the \lcnamecref{MainLemmaForULeftLemma5}.

	Next, come to Item 2.
	We remark that the previous method cannot be applied directly: although on the trunk of $U_2$ we can find points with the same $\psi$-types,
	and we can cut off fragments of $U_2$ between these points,
	it is impossible to ensure that the length of the trunk in the resulting tree will be exactly $k$.
	We thus again start from the smaller tree, $U_1$ (without any fixed set $X_1$ in it), and to each level $i\in[0,N]$ we assign the set of $\psi$-types
	$\Theta_i=\set{\pht{\psi}{U_1\restr{\R^i}}{\varnothing[\X\mapsto X_1]}\mid X_1\subseteq\dom(U_1\restr{\R^i})\land X_1\restr{\R^{k-i}}=X'_1}$
	(here and below we quantify only over finite sets).
	Note that every $\Theta_i$ belongs to $\Pht{\efin\X.\psi}$ (although it is not the $\efin\X.\psi$-type of $U_1\restr{\R^i}$).
	Because $N\geq|\Pht{\efin\X.\psi}|$, we can find a pair of levels $\ell,m$ (with $0\leq\ell<m\leq N\leq k$) for which $\Theta_\ell=\Theta_m$.
	We then construct $U_{1.5}$ from $U_1$ by replicating $\frac{N!}{m-\ell}+1$ times the fragment of $U_1$ lying between the nodes $\R^\ell$ (inclusively) and $\R^m$ (exclusively).
	By \cref{exlemma1} it follows that every $\Theta_i$ is completely determined by $\Theta_{i+1}$ and by the subtree rooted in $\R^i\L$.
	Thus, the replication does not change the set assigned to the root; we have
	\begin{align}
		\set{\pht{\psi}{U_1}{\varnothing[\X\mapsto X_1]}\mid X_1\subseteq\dom(U_1)\land X_1\restr{\R^k}=X'_1}&\nonumber\\
		&\hspace{-12em}=\set{\pht{\psi}{U_{1.5}}{\varnothing[\X\mapsto X_{1.5}]}\mid X_{1.5}\subseteq\dom(U_{1.5})\land X_{1.5}\restr{\R^{k+N!}}=X'_1}.\label{eq:2}
	\end{align}
	
	Having now a set $X_2\subseteq\dom(U_2)$ we transfer it to a set $X_{1.5}\subseteq\dom(U_{1.5})$, so that the $\psi$-type remains the same
	and we have $X_{1.5}\restr{\R^{k+N!}}=X'_1$---%
	in the same way as we previously transferred a set from $U_{1.5}$ to $U_2$.
	Then, by \cref{eq:2} we can find a set $X_1\subseteq\dom(U_1)$, so that the $\psi$-type again remains the same and we have $X_1\restr{\R^k}=X'_1$.
\end{proof}

\begin{observation}\label{MainLemmaForULeftLemma5LeftPath}
	The above lemma holds also when the two directions $\L,\R$ are swapped (then trunks go leftwards);
	the proof is completely symmetric.
\end{observation}

We now have two corollaries.
Recall the paths $S_{m,n}$ from \cref{def:t1t2}.

\begin{corollary}\label{stretch-path}
	Let $\psi$ be such that $|\Pht{\efin\X.\psi}|\leq N$, and let $m,m',n'n'\in\Nat$.
	Then for all $X_1\subseteq\dom(S_{m,n})$ there exists $X_2\subseteq\dom(S_{m',n'})$ such that
	$\pht{\psi}{S_{m,n}}{\varnothing[\X\mapsto X_1]}=\pht{\psi}{S_{m',n'}}{\varnothing[\X\mapsto X_2]}$;
	if $m\leq m'$, then moreover the number of $\a$-labeled elements of $X_2$ is not smaller than the number of $\a$-labeled elements of $X_1$;
	similarly, if $n\leq n'$, then the number of $\b$-labeled elements of $X_2$ is not smaller than the number of $\b$-labeled elements of $X_1$;
	in consequence, if $m\leq m'$ and $n\leq n'$, then $|X_1|\leq|X_2|$.
\end{corollary}

\begin{proof}
	First, we $|n'-n|$ times use the mirrored version of \cref{MainLemmaForULeftLemma5} (cf.~\cref{MainLemmaForULeftLemma5LeftPath})
	to change the length of the bottom $\b$-labeled part of the path from $nN!$ to $n'N!$ (and we use \cref{exlemma1} to see that the type of the whole path remains unchanged);
	we need Item 1 to stretch the path if $n\leq n'$, and Item 2 to contract the path if $n\geq n'$.
	Next, we use the \lcnamecref{MainLemmaForULeftLemma5} another $|m'-m|$ times to change the length of the top $\a$-labeled part of the path from $mN!$ to $m'N!$.
	Note that on the side of the path we have empty subtrees, which clearly all have the same $\efin\X.\psi$-type.
	When $n\leq n'$ and/or $m\leq m'$, we use Item 1, which ensures that the size of the appropriate part of the set does not decrease.
	It is also important that when we stretch/contract one part of the path, potentially increasing the number of its elements in the set,
	then the other part of the set (contained in the rest of the path) remains unchanged.
\end{proof}

\begin{corollary}\label{the-same-efin-type}
	Let $\psi$ be such that $|\Pht{\efin\X.\psi}|\leq N$, and let $m,m',n'n'\in\Nat$.
	Then $\pht{\efin\X.\psi}{S_{m,n}}{\varnothing}=\pht{\efin\X.\psi}{S_{m',n'}}{\varnothing}$.
\end{corollary}

\begin{proof}
	By definition, $\pht{\efin\X.\psi}{S_{m,n}}{\varnothing}$ is the set of $\psi$-types $\pht{\psi}{S_{m,n}}{\varnothing[\X\mapsto X]}$ for all possible $X\subseteq\dom(S_{m,n})$;
	by \cref{stretch-path} this set is the same for all pairs $(m,n)$.
\end{proof}

Next, we use \cref{MainLemmaForULeftLemma5} and its corollaries to prove \cref{lem:T1-T2,lem:T2-T1}.
Recall that they talk about trees $T_1,T_2$ from \cref{def:t1t2}.

\lemTT*
\begin{proof}
	Based on $X_1$ we construct $X_2$.
	First, we choose a number $k'$ such that no elements of the finite set $X_1$ are descendants of $\R^{(k'-k)N!}$, and we take $\ell'=\ell+(k'-k)$.
	We also do not take any descendants of $\R^{(\ell'-\ell)N!}$ to $X_2$.
	Then, by assumptions of our lemma,
	\begin{align}\label{eq:3}
		\pht{\psi}{T_1\restr\R^{k'N!}}{\varnothing[\X\mapsto X_1\restr\R^{(k'-k)N!}]}=\pht{\psi}{T_2\restr\R^{\ell' N!}}{\varnothing[\X\mapsto X_2\restr\R^{(\ell'-\ell)N!}]}.
	\end{align}
	Next, elements of $X_1$ lying outside of vaults are taken to $X_2$ without any change.
	Finally, for every $i\in[0,k'-k-1]$ we construct the part $X_2\restr\R^{iN!}\L$ of $X_2$ (lying in a vault $S_{n,n}$ attached on level $(\ell+i)N!$ in $T_2$)
	using \cref{stretch-path}, so that 
	\begin{align}\label{eq:4}
		\pht{\psi}{T_1\restr\R^{(k+i)N!}\L}{\varnothing[\X\mapsto X_1\restr\R^{iN!}\L]}=\pht{\psi}{T_2\restr\R^{(\ell+i)N!}\L}{\varnothing[\X\mapsto X_2\restr\R^{iN!}\L]}.
	\end{align}
	Because we have equality of $\psi$-types on top of the vaults (\cref{eq:4}) and below the considered fragment of the trees (\cref{eq:3}),
	and because above these places the two trees with the sets in them look the same, by \cref{exlemma1} we obtain that 
	$\pht{\psi}{T_1\restr\R^{kN!}}{\varnothing[\X\mapsto X_1]}=\pht{\psi}{T_2\restr\R^{\ell N!}}{\varnothing[\X\mapsto X_2]}$.
	
	If $k+i\leq 2(\ell+i)$ (i.e., $i\geq k-2\ell$), then the vault $S_{1,m}$ or $S_{m,1}$ starting in the node $\R^{(k+i)N!}\L$ of $T_1$
	and the vault $S_{m',m'}$ starting in the node $\R^{(\ell+i)N!}\L$ of $T_2$ satisfy $m\leq m'$;
	in this case \cref{stretch-path} guarantees that $|X_1\restr\R^{iN!}\L|\leq|X_2\restr\R^{iN!}\L|$.
	When $k\leq 2\ell$, the above covers all $i\geq 0$, so we simply have $|X_1|\leq|X_2|$.
	Otherwise, for $i<k-2\ell$ we have $|X_1\restr\R^{iN!}\L|\leq|\dom(T_1\restr\R^{(k+i)N!}\L)|\leq|\dom(T_1\restr\R^{2(k-\ell)N!}\L)|=|\dom(S_{1,k-\ell})|=(k-\ell+1)\cdot N!$
	and $0\leq|X_2\restr\R^{iN!}\L|$;
	we thus obtain $|X_1|-(k-2\ell)\cdot(k-\ell+1)\cdot N!\leq|X_2|$.
	This gives us the thesis for $f(n)=n-\max(0,k-2\ell)\cdot(k-\ell+1)\cdot N!$.
\end{proof}

\lemTTT*
\begin{proof}
	Suppose first that at least half of elements of $X_2$ lie outside of the vaults.
	We then create $X_1$ based on $X_2$ in the same way as previously (i.e., in the proof of \cref{lem:T1-T2}) we created $X_2$ based on $X_1$.
	This time we do not have any size guarantees for parts of the sets lying in the vaults, but the part of $X_2$ lying outside of the vaults remains unchanged in $X_1$,
	so by our assumption we have $|X_1|\geq f_1(|X_2|)$ for $f_1(n)=\frac{n}{2}$.
	
	We now come to the most interesting case, when at least half of elements of $X_2$ lie inside the vaults.
	We then check which label is more frequent among the vault elements of $X_2$.
	Suppose this is $\a$; then we have at least $\frac{|X_2|}{4}$ $\a$-labeled elements of $X_2$ in the vaults.
	The general idea is to map every vault $S_{i,i}$ to $S_{i,1}$, and adjust all the rest accordingly.
	To this end, recall that $S_{i,i}$ is attached on level $iN!$ in $T_2$, and $S_{i,1}$ is attached on level $2i-1$ in $T_1$.
	We first choose a number $\ell'$ such that no element of $X_2$ is a descendant of $\R^{(\ell'-\ell)N!}$ (i.e., lies in the vault $S_{\ell',\ell'}$ or below).
	Setting $k'=2\ell'-1$, we then also take no descendants of $\R^{(k'-k)N!}$ (i.e., lying in the vault $S_{\ell',1}$ or below) to $X_1$,
	which by assumptions of our lemma implies \cref{eq:3}.
	
	We also set $\ell_0=\max(\ell,\frac{k}{2})+1$, which ensures $\ell_0\geq\ell+1$ and $2\ell_0-1\geq k+1$.
	For every $i\in[\ell_0,\ell'-1]$ we create $X_1\restr\R^{(2i-1-k)N!}\L$ (lying in the vault $S_{i,1}$ of $T_1$)
	based on $X_2\restr\R^{(i-\ell)N!}\L$ (lying in the vault $S_{i,i}$ of $T_2$), using \cref{stretch-path},
	so that
	\begin{align*}
		\pht{\psi}{T_1\restr\R^{(2i-1)N!}\L}{\varnothing[\X\mapsto X_1\restr\R^{(2i-1-k)N!}\L]}=\pht{\psi}{T_2\restr\R^{iN!}\L}{\varnothing[\X\mapsto X_2\restr\R^{(i-\ell)N!}\L]}
	\end{align*}
	and that $|X_1\restr\R^{(2i-1-k)N!}\L|$ is not smaller than the number of $\a$-labeled nodes in $X_2\restr\R^{(i-\ell)N!}\L$.
	In the topmost $\ell_0-\ell$ vaults of $X_2$ we could have at most $(\ell_0-\ell)\cdot\ell_0\cdot N!$ nodes labeled by $\a$,
	so we already have $|X_1|\geq f_2(|X_2|)$ for $f_2(n)=\frac{n}{4}-(\ell_0-\ell)\cdot\ell_0\cdot N!$, no matter how the remaining elements of $X_1$ are selected
	(note that $\ell_0$ does not depend on $X_2$, only on $k$ and $\ell$).
	
	Next, for every $i\in[\ell_0,\ell'-1]$ (formally: starting from the bottom, i.e., from the greatest $i$)
	we use \cref{MainLemmaForULeftLemma5} to choose to $X_1$ some nodes of $T_1$ lying between the vaults $S_{i,1}$ and $S_{i+1,1}$
	(the trunk here has length $2N!-1$, with vault $S_{1,i}$ in the middle) based on the part of $X_2$ between the vaults $S_{i,i}$ and $S_{i+1,i+1}$ in $T_2$
	(where the trunk has length $N!-1$) obtaining
	\begin{align*}
		\pht{\psi}{T_1\restr\R^{(2i-1)N!+1}}{\varnothing[\X\mapsto X_1\restr\R^{(2i-1-k)N!+1}]}=\pht{\psi}{T_2\restr\R^{iN!+1}}{\varnothing[\X\mapsto X_2\restr\R^{(i-\ell)N!+1}]}.
	\end{align*}
	Note that the additional vault $S_{1,i}$ and the non-vault paths $S_{1,1}$ have the same $\efin\X.\psi$-type by \cref{the-same-efin-type}.
	The element just above $S_{i,1}$ in $T_1$ is selected to $X_1$ if and only if the element just above $S_{i,i}$ in $T_2$ was in $X_2$,
	which by \cref{exlemma1} gives us
	\begin{align*}
		\pht{\psi}{T_1\restr\R^{(2i-1)N!}}{\varnothing[\X\mapsto X_1\restr\R^{(2i-1-k)N!}]}=\pht{\psi}{T_2\restr\R^{iN!}}{\varnothing[\X\mapsto X_2\restr\R^{(i-\ell)N!}]},
	\end{align*}
	allowing to proceed by induction on $\ell'-i$.
	
	Finally, we need to handle the top part of the trees.
	We again use \cref{MainLemmaForULeftLemma5} to choose to $X_1$ some elements of $T_1$ above the vault $S_{\ell_0,1}$ (the trunk here has length $(2\ell_0-1-k)N!\geq N!$)
	based on the part of $X_2$ above the vault $S_{\ell_0,\ell_0}$ in $T_2$ (where the trunk has length $(\ell_0-\ell)N!\geq N!$),
	obtaining
	\begin{align*}
		\pht{\psi}{T_1\restr\R^{kN!}}{\varnothing[\X\mapsto X_1]}=\pht{\psi}{T_2\restr\R^{\ell N!}}{\varnothing[\X\mapsto X_2]},
	\end{align*}
	as required.
	Note that the difference between the two trunk lengths may be a multiple of $N!$, and may be positive or negative,
	which means that we may need to apply \cref{MainLemmaForULeftLemma5} several times, using either Item 1 or Item 2.
	
	It remains to handle the case when $X_2$ contains more $\b$-labeled vault elements than $\a$-labeled vault elements.
	This case is analogous: this time we map $S_{i,i}$ into $S_{1,i}$, and we obtain $|X_1|\geq f_3(|X_2|)$ for a function $f_3$ similar to $f_2$.
	
	Collecting the above three cases together, we see that the lemma holds for $f(n)=\min(f_1(n),f_2(n),f_3(n))$.
\end{proof}

\end{document}